\documentclass[a4paper,USenglish]{article}
\usepackage{booktabs} 
\usepackage[linesnumbered,ruled,vlined,ruled]{algorithm2e}

\usepackage[numbers]{natbib}
\usepackage{a4wide}

\usepackage{mathrsfs}

\usepackage{amsmath,graphicx,amsthm,dsfont}

\usepackage{amssymb,amsmath}
\usepackage{multicol}
\usepackage{multirow}
\usepackage{caption}
\usepackage{amsmath}
\usepackage[inline]{enumitem}
\usepackage{siunitx} 
\usepackage{comment} 
\usepackage{amsfonts}
\usepackage{thmtools} 
\usepackage{thm-restate}
\usepackage[pagebackref=True,hidelinks]{hyperref}

\usepackage{dsfont}
\usepackage{csquotes}

\makeatletter
\newif\if@restonecol
\makeatother

\usepackage{algpseudocode}

\newcommand{\cut}[1]{}
\newcommand{\com}[1]{}
\newcommand{\rn}[1]{}

\usepackage[margin=0pt,position=t]{subcaption}
\usepackage{mathtools}
\usepackage[sort&compress,nameinlink,capitalize]{cleveref}
\crefname{table}{Table}{Tables}
\crefname{figure}{Figure}{Figures}
\crefname{cor}{Corollary}{Corollaries}
\crefname{step}{Step}{Steps}
\crefname{rrulev}{Rule}{Rules}
\crefname{thm}{Theorem}{Theorems}
\crefname{obs}{Observation}{Observations}
\crefname{lem}{Lemma}{Lemmas}
\crefname{claim}{Claim}{Claims}
\crefname{section}{Section}{Sections}
\crefname{subsection}{Section}{Sections}
\crefname{figure}{Figure}{Figures}
\crefname{algorithm}{Algorithm}{Algorithms}
\crefname{proposition}{Proposition}{Propositions}
\crefname{theorem}{Theorem}{Theorem}
\crefname{lemma}{Lemma}{Lemmas}
\crefname{construction}{Construction}{Constructions}
\crefalias{AlgoLine}{line}

\newtheorem{theorem}{Theorem}
\newtheorem{definition}{Definition}
\newtheorem{lemma}{Lemma}
\newtheorem{corollary}{Corollary}

\newtheorem{proposition}{Proposition}

\newtheorem{observation}{Observation}

\theoremstyle{remark}

{\bfseries}{\normalfont}
\crefname{rrule}{Rule}{Rules}

\newcommand{\smin}{s_{\text{min}}}
\newcommand{\smax}{s_{\text{max}}}
\newcommand{\MOV}{\text{MOV}}

\newcommand{\cv}{\chi} 
\newcommand{\col}{\text{col}}

\newcommand{\fairCD}{{\normalfont\textsc{FCD}}\xspace}

\newcommand{\pb}{\mathcal{B}} 
\newcommand{\ppb}{\mathcal{D}} 
 
\newcommand{\kp}{\kappa} 
\newcommand{\cc}{\zeta} 

\DeclareMathOperator*{\argmax}{arg\,max}

\DeclareMathOperator{\mln}{mln}
\DeclareMathOperator{\fen}{fen}
\DeclareMathOperator{\fvn}{fvn}
\DeclareMathOperator{\vc}{vcn}
\DeclareMathOperator{\pw}{pw}
\DeclareMathOperator{\twidth}{tw}
\newcommand{\nov}{p}

\usepackage{thmtools} 
\usepackage{thm-restate}

\usepackage{xcolor}
\definecolor{myred}{RGB}{255,153,153}
\definecolor{myorange}{RGB}{255,187,59}
\definecolor{myyellow}{RGB}{255,255,153}
\definecolor{mygreen}{RGB}{153,255,153}

\usepackage{units}
\usepackage[linecolor=green!70!black, backgroundcolor=green!10, bordercolor=black, textsize=scriptsize, obeyFinal
]{todonotes}

\newcommand{\decprob}[3]{%
  \begin{center}%
    \begin{minipage}{0.92\linewidth}%
      \textsc{#1}\\
      \textbf{Input:} #2\\
      \textbf{Question:} #3
    \end{minipage}%
  \end{center}%
}

\tikzstyle{vertex}=[draw, circle, fill, inner sep = 2.4pt]
\tikzstyle{square}=[draw, fill, inner sep = 3pt]
\tikzstyle{penta}=[draw, regular polygon, regular polygon sides=5, fill, inner 
sep = 2.4pt]
\tikzstyle{hexa}=[draw, regular polygon, regular polygon sides=6, fill, inner 
sep = 2.4pt]
\tikzstyle{triangle}=[draw, regular polygon, regular polygon sides=3, fill, inner sep = 2pt]
\tikzstyle{itriangle}=[draw, regular polygon, regular polygon sides=3, 
rotate=180, fill, inner sep = 2pt]

\tikzstyle{star4}=[draw, star, star points=4, fill, inner sep = 2pt]
\tikzstyle{star5}=[draw, star, star points=5, fill, inner sep = 2pt]
\tikzstyle{star6}=[draw, star, star points=6, fill, inner sep = 2pt]

\usetikzlibrary{arrows.meta, shapes.geometric,decorations.pathreplacing}

\usetikzlibrary{shapes,snakes}

\title{A Refined Complexity Analysis of Fair Districting over 
Graphs\footnote{A 2-page extended abstract of this work will appear in the proceedings of the \emph{21st International Conference on Autonomous Agents and Multiagent Systems} (AAMAS 2022).}}

\author{Niclas Boehmer \and Tomohiro 
		Koana \and Rolf Niedermeier}
\date{TU Berlin, Algorithmics and Computational Complexity, 
Germany\\\texttt{\small 
\{niclas.boehmer,tomohiro.koana\}@tu-berlin.de}}

\begin{document}
\allowdisplaybreaks
\maketitle

\begin{abstract}
We study the NP-hard \textsc{Fair Connected Districting} problem recently proposed by Stoica et al.~[AAMAS 2020]: Partition a vertex-colored graph into $k$~connected components (subsequently referred to as districts) so that in every district the most frequent color occurs at most a given number of times more often than the second most frequent color. 
	\textsc{Fair Connected Districting} is motivated by various real-world scenarios where agents of different types, which are one-to-one represented by nodes in a network, 
	have to be  partitioned into disjoint districts.  Herein, one strives for ``fair districts'' without any type being in a dominating majority in any of the districts. This is to e.g.\ prevent segregation or political domination of some political party.  
	We conduct a fine-grained analysis of the (parameterized) computational complexity of \textsc{Fair Connected Districting}. 
	In particular, we 
	prove that it is polynomial-time solvable on paths, cycles, stars, and 
	caterpillars, but already becomes NP-hard on trees.
	Motivated by the latter negative result, we perform a parameterized complexity analysis with respect to various graph parameters, including treewidth, and problem-specific parameters, including, the numbers of colors and districts. 
	We obtain a rich and diverse, close to complete picture 
	of the corresponding parameterized complexity landscape 
	(that is, a classification along the complexity classes FPT, XP, W[1]-hard, and para-NP-hard). 
\end{abstract}

\section{Introduction}
Stoica et al.\ \cite{DBLP:conf/atal/StoicaCDG20} recently introduced 
graph-based problems
on fair (re)districting, employing ``margin of victory'' as the measure of 
 fair representation. In their work, 
they performed theoretical and empirical studies; 
the latter clearly supporting the practical relevance of these problems.
The main contribution of their work is certainly with respect 
to modeling and performing promising empirical 
studies (based on greedy heuristics).
In this paper, we instead focus on the theoretical aspects, significantly extending their
findings in this direction.

Dividing agents into groups is a ubiquitous task.
Electoral districting is one of the prime examples:
Voters are partitioned into voting districts, each electing its own representative.\footnote{One prominent example of electoral districting are the congressional districts in the United States: The US is divided into $435$ congressional districts each electing one member of the House of Representatives. How these districts are drawn is the topic of an ongoing debate \cite{campagna1990party,engstrom2006stacking,hirsch2003united,DBLP:journals/jea/LevinF19}.}
Another example emerges in education; in many countries, 
children are assigned to schools based on their residency.
In such scenarios, the agents (in the settings above, voters or school 
children) are often placed on a (social or geographical) network.
When assigning them to districts, it is natural to require 
that every district should be connected in the network and meet some further criteria.

In districting, there are various objectives.
What we study here can be interpreted as a ``benevolent'' counterpart
of the well-studied gerrymandering scenario in voting theory.
For gerrymandering, every voter is characterized by their 
projected vote in the upcoming election. 
The goal is then to find a partition of the voters into connected 
districts such that some designated alternative gains the majority 
in as many districts as possible. 
Following Stoica et al.\ \cite{DBLP:conf/atal/StoicaCDG20}, we consider an 
opposite objective. 
That is, we assume that some central authority wishes to partition the agents, which are of different types, into connected districts that are \emph{fair}, where a district is deemed fair if the margin of victory in the district is smaller than a given bound. 
The \emph{margin of victory} of a district is the minimum number of agents whose deletion results in a tie between the two most frequent types in the district.
When partitioning children into school districts, types may model sociodemographic attributes such as race and gender, and a low margin of victory could be beneficial to prevent the existence of schools where one trait is in a clear majority and which may thus be only associated with this single trait (see Stoica et al.\ \cite{DBLP:conf/atal/StoicaCDG20} for a more extensive discussion). 
In electoral districting where agents' types can represent their projected vote or ethnicity, a low margin of victory may foster competition among politicians, thereby motivating elected officials to do a great job.
To illustrate that districts that are dominated by a certain ethnicity are a serious problem in particular in developing countries, we quote the two noble price winners Banerjee and Duflo \cite[pp. 251-252]{banerjee2011poor}

\begin{displayquote}
There is reason to be concerned that voting [in developing countries] is often based on ethnic loyalties, which means that the candidate from the largest ethnic group often wins, whatever his intrinsic merit. [...] [I]f voters choose based on ethnicity rather than on merit, the quality of candidates representing the majority group will suffer: These candidates don't need to make much of an effort because the fact that they are from the "right" caste or ethnic group is sufficient to ensure that they are elected. 
\end{displayquote}

Banerjee and Duflo \cite{banerjee2007parochial} even found evidence that in the 1980s and 1990s in North India elected official that belong to the (clearly) dominating caste group were significantly more likely to be corrupt. 
This illustrates the practical importance of creating districts with a low margin of victory (in terms of ethnicity). 

In our work, we build upon the studies of Stoica et 
al.\ \cite{DBLP:conf/atal/StoicaCDG20} to search for tractable special cases 
of fair districting over graphs. 
We focus on the \textsc{Fair Connected Districting} (\fairCD) problem (a natural special case of 
Stoica et 
al.'s \textsc{Fair Connected Regrouping} problem).
The input of \fairCD consists of a graph~$G=(V,E)$ in which every vertex is assigned a color from a set~$C$, and integers $k$, $\ell$, $\smin$, and $\smax$.
The question is whether the vertex set of~$G$ can be partitioned into $k$ connected districts, each containing between $\smin$ and $\smax$ vertices, whose margin of victory is at most~$\ell$. The difference to \textsc{Fair Connected Regrouping} is that \fairCD does not impose any constraints to which districts an agent can be 
assigned.\footnote{We
mention  that Stoica et al.\ \cite{DBLP:conf/atal/StoicaCDG20} use a 
slightly different definition of margin of victory.
While we look at the number of vertices that need to be deleted to have a tied most frequent color, they examine the number of vertices that need to change their color such that the most frequent color changes. 
We chose our definition in order to be able to distinguish the case of two tied most 
frequent colors from the case where one color appears once more than the others 
(which both have margin of victory one in the model of Stoica et 
al.\ \cite{DBLP:conf/atal/StoicaCDG20}). 
Clearly, there are also other fairness measures, e.g., the difference between the occurrences of the most and least frequent color.
While these are natural as well, our definition is particularly appealing if each district may be only associated with its most frequent color if its margin of victory is too high. Notably, margin of victory is also a popular concept in other domains such as group identification \cite{DBLP:conf/ijcai/BoehmerBKL20}, tournament solutions \cite{BRILL2022103600}, and voting theory \cite{DBLP:conf/ijcai/DeyN15,DBLP:conf/sigecom/Xia12}.} 
 
It is easy to see that \fairCD generalizes the known NP-hard \textsc{Perfectly Balanced Connected Partition} problem \cite{DBLP:journals/ipl/Chlebikova96,DBLP:journals/dam/DyerF85}, which asks for a partition of a graph into two connected components of the same size (see \Cref{pr:np-hard}). This motivates a parameterized complexity analysis and the study of restrictions of the underlying graph in order 
to identify tractable special cases. 
Specifically, we analyze the computational complexity of \fairCD on specific graph classes and the parameterized complexity of \fairCD with respect to several problem-specific parameters (such as $ \vert C \vert $ and~$k$) as well as several parameters measuring structural properties of the underlying graph (such as its treewidth or vertex cover number). 

\subsection{Related Work}
Stoica et al.\ \cite{DBLP:conf/atal/StoicaCDG20} introduced
\textsc{Fair Connected Regrouping}, which is a generalization of our \textsc{Fair Connected Districting} problem. 
\textsc{Fair Connected Regrouping} differs from \fairCD in that, in \textsc{Fair Connected Regrouping}, one is additionally given  a function that specifies for each vertex to which district it can belong. 
They proved that \textsc{Fair Connected Regrouping} is NP-hard even for only two colors and two districts. 
Moreover, Stoica et al.\ \cite{DBLP:conf/atal/StoicaCDG20} considered special 
cases of \textsc{Fair Connected Regrouping}:  
\textsc{Fair Regrouping} (omitting connectivity constraints) and \textsc{Fair Regrouping\_X} (omitting connectivity constraints and any restriction to which districts vertices can belong). 
They proved that the former problem is NP-hard for three colors but in XP with respect to the number of districts, while the later problem is additionally in XP with respect to the number of colors. 
In addition, they proposed for each of the variants a greedy heuristic and evaluated them using synthetic and real-world data from UK general elections and US public schools.  

\fairCD is relevant in district-based elections, where voters are partitioned into districts and each district elects its own representative. 
Several papers have studied how to assign voters to districts so as to 
``fairly'' reflect the political choices of voters
\cite{DBLP:conf/ijcai/BachrachLLZ16,DBLP:journals/scw/LandauRY09,DBLP:journals/corr/LandauS14,DBLP:journals/corr/abs-2006-11865,DBLP:journals/mcm/PuppeT08}. 
Well-studied in this context is gerrymandering, which can be regarded 
as a ``malicious'' counterpart to 
our problem. 
In gerrymandering, the task is to partition a set of voters into districts obeying certain conditions such that a designated alternative wins in as many districts as possible. 
An intuitive strategy to solve this problem, which is not necessarily optimal 
\cite{puppe2009optimal}, is to maximize the number of districts where the designated 
alternative wins only by a small margin (this is somewhat related to our problem.)
Initially, gerrymandering has been predominantly studied from the 
perspective of social and political science 
\cite{erikson1972malapportionment,issacharoff2002gerrymandering,lublin1999paradox}. 
More recently, different variants of gerrymandering have been 
considered from an algorithmic perspective 
\cite{DBLP:conf/aaai/EibenFPS20,DBLP:conf/atal/LewenbergLR17}. Notably, the 
study of gerrymandering over graphs, which is analogous to 
our problem, has recently gained significant interest~\cite{DBLP:journals/corr/abs-2102-08905,DBLP:conf/atal/Cohen-ZemachLR18,DBLP:journals/corr/abs-2102-09889,DBLP:conf/atal/ItoK0O19}. In particular, as done here for \fairCD, Bentert et al.\ \cite{DBLP:journals/corr/abs-2102-08905}, Gupta et al.\ \cite{DBLP:journals/corr/abs-2102-09889}, and Ito et al.\ \cite{DBLP:conf/atal/ItoK0O19}  analyzed the complexity of gerrymandering on paths, cycles, and trees and studied the influence of the number of candidates/colors and the number of districts.
 A similar model for graph-based redistribution scenarios and political districting has been studied
 under the name ``network-based vertex dissolution''~\cite{DBLP:journals/siamdm/BevernBCFNW15}.

Partitioning agents of different types into balanced groups is conceptually closely related to studies of social segregation. 
In computer science, social segregation is, for instance, quite extensively studied in the context of Schelling's segregation 
model \cite{schelling1969models}: 
While initially the work in this context was mostly concerned with the 
theoretical analysis of segregation patterns 
\cite{DBLP:conf/stoc/BrandtIKK12,DBLP:conf/soda/BhaktaMR14}, 
recently Schelling's model has been approached from a game-theoretic perspective 
\cite{DBLP:journals/ai/AgarwalEGISV21,DBLP:journals/corr/abs-2105-06561}.

\subsection{Contribution}
Motivated by the NP-hardness of \fairCD (see \Cref{pr:paraNP}), we conduct a parameterized complexity analysis of \fairCD and study restrictions of the underlying graph in order 
to identify tractable special cases. 
We investigate the influence of problem-specific parameters (the number~$ \vert C \vert $ of colors, the number~$k$ of districts, and the margin of victory~$\ell$) and the structure of the underlying graph on the computational complexity of \fairCD.

We show that \fairCD is NP-hard even if $ \vert C \vert =k=2$ and $\ell=0$ but 
polynomial-time solvable on paths, cycles, stars, and caterpillars (for stars, 
our algorithm even runs in linear time).\footnote{While in real-world applications these simple graphs may occur not often, they are the building blocks of more complex graphs. 
This also motivated a study of these graph classes for gerrymandering over graphs \cite{DBLP:journals/corr/abs-2102-08905,DBLP:journals/corr/abs-2102-09889,DBLP:conf/atal/ItoK0O19}. 
Moreover, initially, it is not at all clear that \fairCD is polynomial-time solvable on these graphs, as, for instance, gerrymandering over graphs is NP-hard even on paths \cite[Theorem 1]{DBLP:journals/corr/abs-2102-08905}.
This also proves that \fairCD is sometimes easier than the corresponding gerrymandering over graphs~problem.}  
Subsequently, we extend our 
polynomial-time algorithms for paths and cycles to a polynomial-time algorithm 
for all graphs with a constant max leaf number ($\mln$), which are basically graphs that consist of a constant number of paths and cycles (where the two endpoints of each path and one point from each cycle can be arbitrarily connected). 

Remarkably, in our most involved hardness reduction, we show that \fairCD already becomes  
NP-hard and even W[1]-hard with respect to~$ \vert C \vert +k$ on trees. However, when the
number  of colors or the number of districts is constant, \fairCD on trees becomes polynomial-time solvable. In fact, we show that these results hold for some tree-like graphs as well. Herein, the tree-likeness of a graph is measured by one of three parameters, namely, the treewidth ($\twidth$), the feedback edge number ($\fen$), and the feedback vertex number ($\fvn$).
More precisely, as our most involved algorithmic results, we establish polynomial-time solvability of \fairCD when the number of colors and the treewidth are constant.
We achieve this with a dynamic programming approach on the tree decomposition of the given graph empowered by some structural observations on \fairCD. 
Moreover, we observe that there is a simple 
polynomial-time algorithm on graphs with a constant feedback edge number when there are a constant number of districts.
On the other hand, we prove that \fairCD is 
NP-hard for two districts even on graphs with $\fvn=1$ (and $\twidth=2$).
Lastly, we show that \fairCD is polynomial-time solvable on graphs with a 
constant vertex cover number ($\vc$) and fixed-parameter tractable with respect to the 
vertex cover number and the number of colors.
A summary of our parameterized results can be found in \Cref{fig:params}. Notably, all our hardness results also hold without size~constraints.

In our studies, we identify several sharp complexity dichotomies. 
For instance, \fairCD is polynomial-time solvable on trees with diameter at most 
three but NP-hard and W[1]-hard with respect to~$ \vert C \vert +k$ on 
trees with diameter four. 
Similarly, \fairCD is NP-hard and W[1]-hard with respect to~$ \vert C \vert +k$ on graphs 
with pathwidth at least two but polynomial-time solvable on 
pathwidth-one graphs.	

To summarize, we show that \fairCD without size constraints is NP-hard even in very restricted settings, e.g., on trees or if $ \vert C \vert =k=2$ and $\ell=0$. 
To make the problem tractable, one possibility is to significantly restrict the input graph, e.g., to consist of a constant number of paths and cycles, or to combine structural parameters of the given graph with the number $ \vert C \vert $ of colors or the number $k$ of districts.\footnote{Notably, in most applications, the number of colors and districts should be relatively small. E.g., in their experiments,  Stoica et al. \cite{DBLP:conf/atal/StoicaCDG20} partitioned \num{50000} voters into~$10$~voting districts and \num{41834} schoolchildren into $61$ school districts with $ \vert C \vert =7$. }
For small $ \vert C \vert $ and $k$, the tractability of \fairCD extends to certain tree-like graphs and graphs with a small vertex cover number.  
In contrast to the parameters $ \vert C \vert $ and~$k$, which have a 
strong influence on the complexity of \fairCD, the bound~$\ell$ on the margin of 
victory has only little impact as all hardness results already 
hold for~$\ell=0$ and all our algorithmic results hold for arbitrary~$\ell$. 

\paragraph{Organization} 
The remainder of this paper is structured as follows. 
In \Cref{se:prel}, we formally introduce \fairCD and the graph parameters we examine. 
In \Cref{se:first}, we present some preliminary results on \fairCD mostly concerning NP-hardness. 
In  \Cref{se:paths}, we then consider \fairCD on very special graph classes such as paths and cycles as well as on graphs that 
can be partitioned into a bounded number of those. 
In \Cref{se:trees}, we shift our attention to trees and graphs that are tree-like. 
Lastly, in \Cref{se:vc}, we analyze the influence of the vertex cover number on the complexity of \fairCD.

\begin{figure}[t]
	\centering
	\tikzstyle{arc}=[-{Latex[length=2mm]}]
	\resizebox{0.8\textwidth}{!}{\begin{tikzpicture}[xscale=1.6, yscale=.9]
		\tikzset{
			param/.style={draw, fill=white, rectangle, rounded corners=3, font=\small, minimum width=5em, minimum height=3.8ex},
		}
		\draw[draw=black, fill=myred, rounded corners=3, dotted, line width=0.05cm] (-3, -1) rectangle (2.2, 2.65);
		\draw[draw=black, fill=myorange, rounded corners=3, dashed, line width=0.05cm] (-3, 2.8) rectangle (1.5, 5.35);
		\draw[draw=black, fill=myyellow, rounded corners=3, solid, line width=0.05cm] (1.8, 2.8) rectangle (4.2, 5.35);
		\draw[draw=black, fill=mygreen, rounded corners=3] (3, 4.15) rectangle (4.2, 5.35);
		
		\node[param, inner sep=0pt,inner ysep=+0pt] at (0, -.5) (tw)
		{$\twidth$};
		\node[param, inner sep=+0pt] at (.8, .75) (fvn)
		{$\fvn$};
		\node[param, inner sep=+0pt] at (1.6, 2) (fen)
		{$\fen$};
		\node[param, inner sep=+0pt] at (-.8, .75) (twk)
		{$\twidth + \, k$};
		\node[param, inner sep=+0pt] at (0, 2) (fvnk)
		{$\fvn + \, k$};
		\node[param, inner sep=+0pt] at (-2.2, .8) (ck)
		{$ \vert C \vert  + k$};
		
		\node[param, inner sep=+0pt] at (-1.5, 3.5) (twc) {$\twidth + \,  \vert C \vert $};
		\node[param, inner sep=+4pt] at (-.8, 4.75) (twck) {$\twidth + \,  \vert C \vert  + k$};
		\node[param, inner sep=+0pt] at (.9, 3.5) (fenk) {$\fen + \, k$};
		
		\node[param, inner sep=+0pt] at (2.4, 3.5) (mln) {$\mln$};
		\node[param, inner sep=+0pt] at (3.6, 3.5) (vc) {$\vc$};
		
		\node[param, inner sep=+0pt] at (3.6, 4.75) (vcc) {$\vc + \,  \vert C \vert $}; 
		
		\draw[arc, rounded corners=10] (ck)  |- (twck);
		\draw[arc] (tw) -- (twk);
		\draw[arc] (twk) -- (fvnk);
		\draw[arc] (fvnk) -- (fenk);
		\draw[arc] (fen)-- (fenk);
		\draw[arc] (tw) -- (fvn);
		\draw[arc] (fvn) -- (fen);
		\draw[arc] (fvn) -- (fvnk);
		\draw[arc] (fen) -- (fenk);
		\draw[arc] (twk) -- (twck);
		\draw[arc] (twk) -- (fvnk);
		\draw[arc, rounded corners=10] (tw) -|  (twc);
		\draw[arc] (twc) -- (twck);
		\draw[arc] (fen) -- (mln);
		\draw[arc, rounded corners=10] (fvn) -|  (vc);
		\draw[arc] (vc) -- (vcc);
		\end{tikzpicture}}
	\caption{
		Overview of our parameterized complexity results.
		Each box represents one parameterization of \fairCD.
		An arc from parameter~$p$ to another parameter~$p'$ indicates that $p$ is upper-bounded by some function of~$p'$.
		For parameters in the red area (dotted), we prove that \fairCD is NP-hard even if the parameter is a constant.
		For parameters in the orange area (dashed), we prove W[1]-hardness and present an XP-algorithm.
		For parameters in the yellow area (solid thick), we have an XP-algorithm but W[1]-hardness is unknown.
		The green area (solid) indicates fixed-parameter tractability.
	}
	\label{fig:params}
\end{figure}
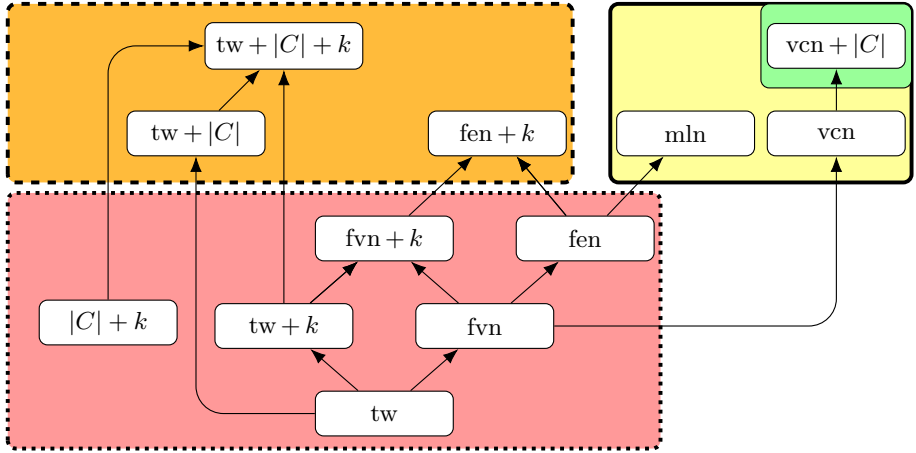

\section{Preliminaries} \label{se:prel}
For $a,b\in \mathbb{N}$, let $[a,b]$ denote $\{a,a+1,\dots, b-1,b\}$ and let $[b]$ denote $[1,b]$. Throughout the paper, all graphs $G = (V, E)$ are undirected and have no self-loops or multi-edges. By convention, we will use $n:=  \vert V \vert $.
Given a graph $G=(V,E)$ and a vertex set $V'\subseteq V$, let $G[V']$ be the graph $G$ induced by the vertices from $V'$. 

\paragraph{\textsc{Fair Connected Districting}.} 
Let $C=\{c_1,\dots, c_{ \vert C \vert }\}$ be the set of colors. 
We assume that each vertex $v\in V$ of the given graph has a color $c\in C$ determined by a given coloring function $\col:V \rightarrow C$. 
For a vertex set $V'\subseteq V$, let $\cv_c(V')$ denote the number of vertices of color $c$ in $V'$. 
Moreover, let $\cv(V') \in \mathbb{N}^{ \vert C \vert }$ denote the vector in which the $i$th entry contains the number of vertices of color $c_i$ in $V'$, i.e, $\cv_i(V')=\cv_{c_i}(V')$.
For a vector $\mathbf{x}=(x_1,\dots, x_t) \in \mathbb{N}^{t}$, let $i^*_\mathbf{x}$ denote an index with the largest entry in $\mathbf{x}$, i.e., $i^*_\mathbf{x}\in \argmax_{i\in [t]} x_i$.
The \emph{margin of victory} $\MOV(\mathbf{x})$ of a vector $\mathbf{x}$ is defined as the difference between the largest and second largest entry in $\mathbf{x}$, i.e., $\MOV(\mathbf{x}):=x_{i^*_\mathbf{x}}- \max_{i\in [t]\setminus \{i^*_\mathbf{x}\}} x_i$. 
If $t=1$, then we set $\MOV(\mathbf{x}):=  x_1$.
Accordingly, we define the margin of victory $\MOV(V')$ of a vertex set $V'\subseteq V$ of colored vertices as $\MOV(V'):=\MOV(\cv(V'))$.
For $\ell\in \mathbb{N}$, we call a vertex set $V'$ \emph{$\ell$-fair} if $\MOV(V')\leq \ell$. We use the term \emph{district} to refer to a vertex set $V'\subseteq V$. We say that a color $c_i\in C$ is the $j$th most frequent color in a district $V'$ if $\cv_i(V')$ is the $j$th~largest entry in~$\cv(V')$ (we break ties arbitrarily unless stated~otherwise). 
We now present our central problem:\footnote{Note that we chose to require that all districts are non-empty to be consistent with the closely related literature on gerrymandering over graphs \cite{DBLP:conf/atal/Cohen-ZemachLR18,DBLP:conf/atal/ItoK0O19}. Moreover, our variant can easily be used to decide whether there exists a partitioning into at most $k$ non-empty districts by simply iterating over all $i\in [k]$. For the reverse direction, this is not easily possible. For instance, consider a clique with two blue vertices and one green and one red vertex and $\ell=0$. Then, a partitioning into two non-empty $\ell$-fair districts is possible but not into any other number of non-empty districts.}
 
\decprob{Fair Connected Districting (\fairCD)}
{A graph $G=(V,E)$, a set $C$ of colors, a function $\col:V\rightarrow C$ assigning each vertex one color from $C$, a number $k\leq  \vert V \vert $ of districts, a maximum margin of victory $\ell$, and two integers $\smax \ge \smin \ge 1$.}
{Does there exist a partition of the vertices into $k$ districts $(V_1,\dots,V_k)$ such that, for all $i\in [k]$, $V_i$ is $\ell$-fair, $ \vert V_i \vert \in [\smin, \smax]$, and $G[V_i]$ is connected?}

\paragraph{Graph Parameters.}
We define several graph parameters for an undirected graph $G=(V,E)$.
The \emph{diameter} of~$G$ is the maximum shortest distance between any pair of vertices.
The \emph{max leaf number} $\mln(G)$ (for a connected graph~$G$) is the maximum number of leaves over all spanning trees of $G$.
A vertex set $V' \subseteq V$ is a \emph{vertex cover} of~$G$ if $G[V \setminus V']$ has no edge.
A vertex set $V' \subseteq V$ is a \emph{feedback vertex set} if $G[V \setminus V']$ is a forest.
Analogously, an edge set $E' \subseteq E$ is a \emph{feedback edge set} if $(V, E \setminus E')$ is a forest.
The \emph{vertex cover number} $\vc(G)$, \emph{feedback vertex number} $\fvn(G)$, and \emph{feedback edge number} $\fen(G)$  is the size of a smallest vertex cover, feedback vertex set, and feedback edge set, respectively.
A \emph{tree decomposition} of a graph $G = (V, E)$ is a pair $(T, \{ B_x \}_{x \in V_T})$, where $T = (V_T, E_T)$ is a rooted tree and $B_x \subseteq V$ for each $x \in V_T$ such that 
\begin{enumerate}[label=(\roman*)]
  \item
    $\bigcup_{x \in V_T} B_x = V$,
  \item
    for each edge $\{u, v\} \in E$, there is an $x \in V_T$ with $u, v \in B_x$, and 
  \item
    for each $v \in V$, the set of nodes $x \in V_T$ with $v \in B_x$ induces a connected subtree in $T$.
\end{enumerate}
The \emph{width} of $(T, \{ B_x \}_{x \in V_T})$ is $\max_{x \in V_T}  \vert B_x \vert  - 1$.
The \emph{treewidth} $\twidth(G)$ of $G$ is the minimum width of all tree decompositions of $G$.
The \emph{pathwidth} $\pw(G)$ of $G$ is defined analogously with the additional constraint that $T$ is a path.

If $G$ is clear from context, then we simply omit it and write $\mln$, $\vc$, $\fvn$, $\fen$, $\twidth$, and $\pw$, respectively. 
For all our parameterized algorithms with respect to one of these parameters, we assume that the corresponding structure is given as part of the input. For instance, in our algorithm parameterized by $\twidth$, we assume that we are given a tree decomposition. 
However, note that in all cases, the worst-case running time of our algorithms would not change if we compute the structure using known parameterized algorithms. 

 \paragraph{Parameterized Complexity Theory} \label{sec:PC}
 A \emph{parameterized problem}~$L$ consists of a problem instance~$\mathcal{I}$ and a parameter value~$k\in \mathbb{N}$. 
Then $L$ lies in XP with respect to $k$ if there exists an algorithm deciding~$L$ in $ \vert \mathcal{I} \vert ^{f(k)}$ time for some computable function $f$. 
Furthermore, $L$~is called \emph{fixed-parameter tractable} with respect to $k$ if there exists an algorithm deciding~$L$ in $f(k) \vert \mathcal{I} \vert ^{\mathcal{O}(1)}$ time for a computable function~$f$. The corresponding complexity class is called~FPT.
There is a hierarchy of complexity classes for parameterized problems: FPT$\subseteq$ W[1] $\subseteq$ W[2] $\subseteq$ XP, where it is commonly believed that all inclusions are strict. 
Thus, if~$L$ is shown to be W[1]-hard, then the common belief is that it is
not fixed-parameter tractable. For instance, 
computing a clique of size at least~$k$ in a graph is known to be W[1]-hard with respect to the parameter $k$.
One can show that $L$ is W[$t$]-hard, $t\geq 1$, by a \emph{parameterized reduction} from a known W[$t$]-hard parameterized problem~$L'$. 
A parameterized reduction from $L'$ to $L$ is a function that maps an instance $(\mathcal{I}', k')$ of $L'$ to an instance $(\mathcal{I}, k)$ of $L$ such that $(\mathcal{I}', k')$ is a yes-instance for $L'$ if and only if $(\mathcal{I}, k)$ is a yes-instance for $L$. 
Moreover, we require that $k$ is bounded in a function of $k'$ and that the  transformation takes at most $f(k') \vert \mathcal{I} \vert ^{\mathcal{O}(1)}$ time for some 
computable function $f$.
Finally, we say that a parameterized problem is \emph{para-NP-hard} if it is NP-hard even for constant parameter values.

\section{Basic Results on \textsc{FCD}} \label{se:first}
In this section, we make some basic observations on the computational complexity of \fairCD. 
First, we prove that \fairCD is para-NP-hard with respect to the combination $ \vert C \vert +k+\ell$ of all three problem-specific parameters. 
This strong general hardness result motivates the study of various restricted graph classes in subsequent sections.

The NP-hardness of \fairCD easily follows from the fact that it generalizes the \textsc{Perfectly Balanced Connected Partition} problem, which is NP-hard on bipartite graphs \cite{DBLP:journals/dam/DyerF85}: 
\decprob{Perfectly Balanced Connected Partition}
{A graph $G=(V,E)$.}
{Is there a partition $(V_1,V_2)$ of the vertex set~$V$ such that $G[V_1]$ and~$G[V_2]$ are connected and $ \vert V_1 \vert  =  \vert V_2 \vert $?
}

Consider the following polynomial-time many-one reduction:
Given an instance $G=(V,E)$ of \textsc{Perfectly Balanced Connected Partition}, construct an equivalent instance of \fairCD by coloring all vertices in $G$ in the same color, setting $k=2$, $\ell= \vert V \vert /2$, $\smin=1$, and $\smax=\infty$.
Moreover, it is also possible to prove the NP-hardness for the case with $\ell=0$ at the cost of introducing a second color:
\begin{proposition}
\label{pr:paraNP}\label{pr:np-hard}
 \fairCD is NP-hard even if $G$ is bipartite, $\smin=1$, $\smax=\infty$, and (i) $ \vert C \vert =1$, $k=2$ or (ii) $\ell=0$, $ \vert C \vert =2$,~$k=2$.
\end{proposition}
\begin{proof}
 To prove (ii), we present a polynomial-time Turing reduction from \textsc{Perfectly Balanced Connected Partition} to \fairCD. 
 To this end, we assume that we have access to an oracle for \fairCD. 
 
 \noindent \textbf{Construction.} Given an instance $(G=(V,E))$ of \textsc{Perfectly Balanced Connected Partition} with $n:= \vert V \vert $, for each pair of vertices $(v, v')\in V^2$, we construct an instance of \fairCD: 
 We set $C=\{c_1,c_2\}$, $k=2$, and $\ell=0$. 
 We color all vertices of the given graph~$G$ in color~$c_1$. 
 Moreover, we modify $G$ by introducing $n$ new vertices of color $c_2$. 
 We connect half of these vertices to $v$ and the other half to $v'$. 
 Subsequently, we query our oracle for \fairCD on the resulting instance and return yes if the oracle returns yes; otherwise, we continue with the next pair of vertices.
 After the oracle rejected all instances, we return no. 
 
 We now prove the correctness of the reduction by showing that there exists a pair of vertices $(v,v')\in V^2$ for which the constructed \fairCD instance is a yes-instance if and only if the given \textsc{Perfectly Balanced Connected Partition} instance is a yes-instance. 
 
 Let $(V_1,V_2)$ be a solution to the \fairCD instance constructed for some pair of vertices $(v,v')\in V^2$. 
 As $\ell=0$ and we require that both $V_1$ and $V_2$ need to be non-empty, it needs to hold that both districts contain vertices of both colors. 
 Moreover, as both $G[V_1]$ and $G[V_2]$ are connected, all vertices of color $c_2$ adjacent to $v$ need to be in the same district as $v$ (and similarly for $v'$).
 Thus, $v$ and $v'$ need to be in different districts, each consisting of $\frac{n}{2}$ vertices of color $c_2$ and $\frac{n}{2}$ vertices of color $c_1$.
 Thereby, after removing all vertices of color $c_2$, $(V_1,V_2)$ is a solution to the given \textsc{Perfectly Balanced Connected Partition} instance.
 
 To prove the reverse direction, let $(V_1,V_2)$ be a solution to the given \textsc{Perfectly Balanced Connected Partition} instance. 
 Let $v\in V_1$ be some vertex from $V_1$ and $v'\in V_2$ be some vertex from $V_2$. 
 Then, the \fairCD instance constructed for $(v,v')$ is a yes-instance, as it 
 is possible to add the $\frac{n}{2}$ vertices of color $c_2$ attached to $v$ 
 to $V_1$ and the $\frac{n}{2}$ vertices of color $c_2$ attached to $v'$ to 
 $V_2$ to arrive at a solution to the \fairCD instance constructed for pair 
 $(v,v')\in V^2$.
\end{proof}
Note that our results strengthen a result of Stoica et 
al.\ \cite{DBLP:conf/atal/StoicaCDG20}, who proved that \fairCD is NP-hard for $ \vert C \vert =2$ and $k=2$ if we can additionally specify for each vertex the districts to which it can be assigned. 

As our second basic result, using a simple dynamic programming approach, we show that an instance of \fairCD{} on a disconnected graph is polynomial-time solvable if \fairCD{} can be solved in polynomial time on its connected components.
This will play a crucial role, e.g., in developing an XP-algorithm for the max leaf number (\Cref{th:mln}). 

\begin{proposition}
\label{prop:disjointunion}
Let $\mathcal{G}$ be a class of graphs such that $\fairCD$ is polynomial-time solvable on any graph $G \in \mathcal{G}$.
  Then, \fairCD is polynomial-time solvable on any graph from $\mathcal{G}'$, where $\mathcal{G'}$ is the class of graphs obtained by taking disjoint unions of graphs from $\mathcal{G}$.
\end{proposition}
\begin{proof}
  Let $G'\in \mathcal{G}'$ and let $G_1, \dots, G_p\in \mathcal{G}$ be the connected components of $G'$. 
  Clearly, every district is contained in the vertex set of $G_i$ for some $i \in [p]$.
  We solve the problem using a simple subset-sum like dynamic programming algorithm. To this end, we introduce a table $T[i,j]$ for $i\in [p]$ and $j\in [k]$. 
  An entry~$T[i,j]$ is true if one can partition the vertices of $G_1, \dots, G_i$ into $j$ connected $\ell$-fair districts respecting the size constraints. 
  We also use a table $H$ where $H[i,j]$ for $i\in [p]$ and $j\in [k]$ is true if the vertex set of $G_i$ can be partitioned into $j$ connected $\ell$-fair districts respecting the size constraints.
  Note that $H$ can be computed in polynomial time by our initial assumption on $G$.

  To initialize $T$, we set $T[1,j]$ to $H[1,j]$ for all $j\in [k]$. 
  Subsequently, for increasing $i>1$, we update $T$ as follows: 
  $$T[i,j]=\bigvee_{\substack{j',j''\in [j]\\ j'+j''=j}} T[i-1,j'] \wedge H[i,j''].$$ 
  In the end, we return $T[p,k]$. 
  This algorithm runs in $\mathcal{O}(p\cdot k^2) = \mathcal{O}(n^3)$ time 
  (aside from the computation of $H$).
\end{proof}

\section{\textsc{FCD} on Paths, Cycles, and Beyond} \label{se:paths}
This section studies the computational complexity of \fairCD on simple graphs and graphs that can be partitioned into ``few'' paths and cycles. 
Specifically, in \Cref{sub:simplegraphs}, we develop polynomial-time algorithms on paths, cycles, stars and caterpillars. 
Subsequently, in \Cref{sub:mln}, we show that \fairCD is in XP when parameterized by the max leaf number, which generalizes polynomial-time solvability on paths and cycles.

\subsection{Polynomial-Time Algorithms for \textsc{FCD} on Simple Graph Classes} \label{sub:simplegraphs}
We start by proving that \fairCD is cubic-time solvable on paths using a simple dynamic programming approach.
\begin{proposition}
  \label{prop:path}
  \fairCD on paths can be solved in $\mathcal{O}(k \cdot n^2)$ time.
\end{proposition}
\begin{proof}
  Let $G=(\{v_1,\dots,v_n\}, \{\{v_i,v_{i+1} \} \mid i\in [n-1]\})$ be the input path. 
  We first create a table~$A[i,j]$, where $A[i,j]$ for $i\leq j\in [n]$ is true if and only if $\{v_i,\dots , v_j\}$ is $\ell$-fair and $ \vert \{v_i,\dots, v_j\} \vert =j - i + 1 \in [\smin, \smax]$.
  We then create a table~$T$ with entries $T[i,t]$ for $i\in [n]$ and $t\in [k]$.
  The meaning of an entry $T[i,t]$ is that it is true if there is a partition of the vertices $\{v_1,\dots, v_i\}$ into $t$ paths (connected districts) that are all $\ell$-fair and respect the size constraints. 
  Note that the input is a yes-instance if and only if $T[n,k]$ is true.

  We initialize table~$T$ by setting $T[i,1]$ for $i\in [n]$ to $A[1,i]$. 
  Subsequently, for $t>1$, we update the table using:
  $$T[i,t]=\bigvee_{j\in [i-1]} T[j,t-1] \wedge A[j+1,i].$$
  The reasoning behind this is that if we partition $\{v_1,\dots , v_i\}$ into $t$ districts, then there needs to be some $j\in [i-1]$ such that $\{v_{j+1},v_{j+2},\dots, v_i\}$ is one district in the solution. 
  Thus, if $T[i,t]$ is true, then there needs to be some $j\in [i-1]$ such that the vertices from $\{v_1,\dots, v_j\}$ can be partioned into $t-1$ $\ell$-fair districts respecting the size constraints and the vertices $\{v_{j+1},v_{j+2},\dots, v_i\}$ form an $\ell$-fair district respecting the size constraints. 

  The table $A$ can be filled in $\mathcal{O}(n^2)$ time: 
  for a fixed $i\in [n]$, we can compute all $A[i,n-j]$ for $j\in [0,n-i]$ in linear time by starting for $j=0$ with computing $\cv(\{v_i,\dots, v_{n-j}\})$ and, subsequently, for increasing~$j>1$ compute $\cv(\{v_i,\dots, v_{n-j}\})=\cv(\{v_i,\dots, v_{n-{j+1}}\})- \cv(\{v_{n-j+1}\})$ (which can be done in constant time).  
  It can be determined in $O(1)$ time whether $\cv(\{v_i,\dots, v_{n-j}\})$ is $\ell$-fair by additionally storing the number of occurrences of integers (which are at most $n$) in $\cv(\{v_i,\dots, v_{n-j}\})$.
  Additional bookmarking allows us to find the two largest entries in $\cv(\{v_i,\dots, v_{n-j}\})$ in $\mathcal{O}(1)$ time.
  Moreover, the number of occurrences in $\cv(\{v_i,\dots, v_{n-j}\})$ can be updated in $\mathcal{O}(1)$ time as we increase $j$.
  We set $A[i, n - j]$ to true if and only if $\cv(\{v_i,\dots, v_{n-j}\})$ is $\ell$-fair and $\smin \le n - i - j + 1 \le \smax$.
 
  Note that table $T$ consists of $k\cdot n$ entries. 
  We spend $\mathcal{O}(n)$ time for every entry, resulting in an overall 
  running time of $\mathcal{O}(k\cdot n^2)$.
\end{proof}

To solve \fairCD on cycles, we iterate over all vertices of the cycle as the starting point of the first district and split the cycle at this point to convert it into a path. Subsequently, we employ the algorithm from above, which results in a running time of $\mathcal{O}(k\cdot n^3)$: 
\begin{corollary}
  \fairCD on cycles is solvable in $\mathcal{O}(k\cdot n^3)$ time.
\end{corollary}

Next, we proceed to stars, for which we derive a precise characterization of yes-instances.
In fact, using a relatively involved analysis, we consider a more general problem, in which some set $X$ of leaves must belong to the same district as the center vertex, i.e., the vertex which is adjacent to all other vertices from the star. 
This proves useful in speeding up the algorithms to be presented 
in \Cref{th:cater,pr:v2}.

\begin{proposition}
\label{le:stars}
 Let $(G=(X\cup Y,E),C,\col,k,\ell,\smin,\smax)$ be an \fairCD instance where $G$ is a star with center vertex $v\in X$. Then, one can decide in linear time whether there is a partition of $X\cup Y$ into $k$~connected $\ell$-fair districts respecting $\smin$ and $\smax$ such that all vertices from $X$ are part of the same district. 
\end{proposition}
\begin{proof}
 In the following, we call a partition of the vertices in $X\cup Y$ into $k$~connected districts a solution if each district is $\ell$-fair, respects the size constraints, and all vertices from $X$ are part of the same district to which we refer as the center district. 
 Note that each solution consists of the center district and several districts consisting of a single vertex from~$Y$.
 Hence, we may assume that $\smin = 1$.
 (For $\smin \ge 2$, we can immediately conclude that there is no solution unless $k = 1$, $X \cup Y$ is $\ell$-fair, and $ \vert X\cup Y \vert \in [\smin,\smax]$.)
 We may also assume that $ \vert X \vert  \le \smax$.
 We now describe a linear-time algorithm that decides whether a given \fairCD instance admits a solution. 
 
 We start with several special cases: For $\ell=0$, a solution exists if and only if $\MOV(X\cup Y) = 0$, $k = 1$, and $ \vert X\cup Y \vert \leq \smax$.
 If $ \vert C \vert =1$, then all solutions containing $k$ districts are isomorphic, i.e., they contain $k-1$ districts containing a single vertex and a single district containing the remaining vertices. 
 Thus, the given instance admits a solution if and only if $ \vert Y \vert \geq k-1$,  $\ell\geq \cv_{c}(Y\cup X)-(k-1)$ where $C = \{ c \}$, and $ \vert Y\cup X \vert -(k-1)\leq \smax$.  
 
 Let $c_1\in C$ be the most frequent color in $X\cup Y$ and $c_2\in C$ be the second most frequent color in $X\cup Y$.
 If $\cv_{c_1}(X) > \cv_{c_2}(X \cup Y) + \ell$, then the given instance has no solution because the center district cannot be $\ell$-fair. 
 
 Having dealt with several special cases separately, we now present a precise characterization of yes-instances assuming that $ \vert C \vert >1$, $\ell>0$, $\smin = 1$, $\smax \le  \vert X \vert $, and $\cv_{c_1}(X) \le \cv_{c_2}(X \cup Y) + \ell$.
 Let $c^{\star}_1$ be the most frequent color in $X$. 
 Further, let $c^{\star}_2 \neq c^{\star}_1$ be the most frequent color in $X$ among colors $c'$ such that $\cv_{c'}(X\cup Y)\geq \cv_{c^{\star}_1}(X)-\ell$. 
 Note that $c^{\star}_2$ is always well-defined. 
 If $c^{\star}_1\neq c_1$, then by the definition of~$c_1$ it holds that $\cv_{c_1}(X\cup Y)\geq \cv_{c^{\star}_1}(X\cup Y) > \cv_{c^{\star}_1}(X)-\ell$. 
 On the contrary, if $c^{\star}_1=c_1$, then, by our assumption that $\cv_{c_1}(X) \leq \cv_{c_2}(X \cup Y) + \ell$, we have $\cv_{c_2}(X\cup Y)\geq \cv_{c^{\star}_1}(X)-\ell$.
 
 Let $b_l:=\max(0, \cv_{c_1}(X \cup Y)- \cv_{c_2}(X\cup Y)- \ell)$ and $b_u:=\max(0, \cv_{c^{\star}_1}(X)-\ell- \cv_{c^{\star}_2}(X))$.
 Then, there is a solution if and only if $k\in [b_l+1,  \vert Y \vert +1-b_u]$ and $k >  \vert X \vert  +  \vert Y \vert  - \smax$.
 
 We first prove the forward direction, i.e., that the existence of a solution implies that $k\in [b_l+1,  \vert Y \vert +1-b_u]$ and $k>  \vert X \vert + \vert Y \vert -\smax$. 
 For any solution, the center district contains at most $\cv_{c_2}(X \cup Y) + \ell$ vertices of color $c_1$ (otherwise it cannot be $\ell$-fair).
 Thus, at least $\cv_{c_1}(X \cup Y) - \cv_{c_2}(X \cup Y) - \ell$ vertices of color $c_1$ have to be part of a non-center district, which implies that $k > b_l$.
 Next, we show that $k \le  \vert Y \vert  + 1 - b_u$.
 The center district has at least $\cv_{c^{\star}_1}(X)$ vertices of color~$c^{\star}_1$.
 Since the center district is $\ell$-fair, it has at least $\cv_{c^{\star}_1}(X) -\ell$ vertices of color $c$ for some color $c\in C \setminus \{c^{\star}_1\}$. 
 As only $\chi_c(X)$ of the  $\chi_{c_1^{\star}}(X)-\ell$ vertices come from $X$, it follows that there are at least $\cv_{c^{\star}_1}(X)- \ell - \cv_{c}(X) $ vertices of color $c$ from $Y$ that are included in the center district. 
 As $c_2^{\star}$ is the most frequent color in $X$ among all colors occurring at least $c_{c_1^{\star}}(X)-\ell$ times, it follows that $\chi_{c_1^{\star}}(X)-\ell-\chi_{c_2^{\star}}(X)$ minimizes the expression from the previous sentence. 
 Using this, we can conclude that $k \le  \vert Y \vert  + 1 - b_u$.
 Finally, we show that $k >  \vert X \vert  +  \vert Y \vert  - \smax$.
 Since the center district contains at most $\smax$ vertices, at least $ \vert X \vert  +  \vert Y \vert  - \smax$ vertices are not part of the center district, which implies that $k >  \vert X \vert  +  \vert Y \vert  - \smax$.
 
 We now prove the backward direction, i.e., that $k\in [b_l+1,  \vert Y \vert +1-b_u]$ and $k >  \vert X \vert  +  \vert Y \vert  - \smax$ imply the existence of a solution.
 First, we show that if $k \in [b_l,  \vert Y \vert  + 1 - b_u]$, ignoring the size constraints, a solution always exists.
 We do this by first constructing two  ``extreme'' solutions where as many or as few vertices as possible are contained in the center district.
 Recall that by our assumption $\cv_{c_1}(X) \le \cv_{c_2}(X \cup Y) + \ell$.
 So we have $\cv_{c_1}(Y) = \cv_{c_1}(X \cup Y) - \cv_{c_1}(X) \ge \cv_{c_1}(X \cup Y) - \cv_{c_2}(X \cup Y) - \ell$, and hence $\cv_{c_1}(Y) \ge b_l$. 
 A solution consisting of $b_l+1$ districts exists: 
 Put all vertices from~$X\cup Y$ in one district and move $b_l$ vertices of color~$c_1$ from~$Y$ each in their own district (this is always possible by our above observation).
 Let $S$ be the center district in this solution. 
 $S$ is $\ell$-fair, as $\chi_{c_1}(S)\geq \chi_{c_2}(S)$, $\chi_{c_1}(S)=\chi_{c_1}(X\cup Y)-b_l\leq \chi_{c_2}(X\cup Y)-\ell=\chi_{c_2}(S)-\ell$, and $\chi_{c_2}(S)\geq \chi_c(S)$ for all $c\in C\setminus \{c_1\}$. 
 A solution consisting of $ \vert Y \vert +1-b_u$ districts is also guaranteed to exist: 
 Put the vertices of $X$ and $b_u$ vertices of color $c^{\star}_2$ from $Y$ into the center district; put all remaining vertices from $Y$ into their own district.
 Let $S'$ be the center district in this solution.
 $S'$ is $\ell$-fair, as $\chi_{c_1^{\star}}(S)-\ell= \chi_{c_1^{\star}}(X)-\ell\leq\chi_{c_2^{\star}}(X)+b_u=\chi_{c_2^{\star}}(S)\leq \chi_{c_1^{\star}}(S)$. 
 Note that from these two solutions, we can construct a solution consisting of $k$ districts for arbitrary $k\in [b_l+1,  \vert Y \vert +1-b_u]$. 
 For this, we start with the first solution and one by one move vertices from $Y \cap (S\setminus S')$ from the center district into their own district while maintaining that the center district is $\ell$-fair. 
 Such a vertex always exists:
  If there is a vertex $v$ of the most frequent color of the center district in $Y \cap (S\setminus S')$, then we move $v$ from $S$ into its own district (removing $v$ only increases the margin of victory by one if the center district is currently $0$-fair and we have $\ell>0$).
  Otherwise, we remove the remaining vertices from $Y \cap (S\setminus S')$ in 
  an arbitrary order and put each into its own district (doing so results in an $\ell$-fair center district, since 
  otherwise $S'$ would not be $\ell$-fair).
  If $k >  \vert X \vert  +  \vert Y \vert  - \smax$, then the above constructed solution for this $k$ satisfies the size constraints as it consists of a center district containing $ \vert X \vert + \vert Y \vert -(k-1)\leq \smax$ vertices where the inequality holds by our observation on $k$ .
\end{proof}

\Cref{le:stars} directly implies that \fairCD is linear-time solvable on stars, even if  the center vertex has a weight for each color, i.e., the center vertex represents multiple vertices.
\begin{corollary}
 \fairCD is linear-time solvable on stars.
\end{corollary}

Lastly, we extend our polynomial-time algorithm for paths from \Cref{prop:path} to caterpillar graphs. A caterpillar graph is a tree where every vertex is either on a central path (spine) or a neighbor of a vertex on the central path. 
\begin{proposition} \label{th:cater}
  \fairCD on caterpillars can be solved in $\mathcal{O}(k\cdot n^3 )$ time.
\end{proposition}
\begin{proof}
  Let $G=(V,E)$ be the given caterpillar, let $(u_1, \dots, u_p)$ denote the spine (central path) of~$G$ and let $U=\{u_1,\dots , u_p\}$. 
  Moreover, for $i\in [p]$, let $U_i$ consist of $u_i$ and all vertices from $V\setminus U$ adjacent to $u_i$. 
  Note that each vertex from $V\setminus U$ is only adjacent to one vertex from $U$, as $G$ is in particular a tree. 
  We solve the problem by extending our algorithm for paths from \Cref{prop:path}. 
  
  As a first step, we introduce a table $A[i,j,t]$ for $i\leq j\in [n]$ and $t\in [k]$. 
  An entry $A[i,j,t]$ is set to true if there is a partition $(V_1, \dots, V_t)$  of $\bigcup_{i'\in [i,j]} U_{i'}$ into $t$ districts such that: 
  \begin{itemize}
   \item $V_1$ contains $u_{i'}$ for every $i' \in [i, j]$, and
   \item  $V_{t'}$ is $\ell$-fair, $ \vert V_{t'} \vert  \in [\smin, \smax]$, and $G[V_{t'}]$ is connected for every $t' \in [t]$.
  \end{itemize}
  We can fill table $A$ in $\mathcal{O}(k\cdot n^3)$ time using \Cref{le:stars} for each $A[i,j,t]$ (for this, it is necessary to slightly restructure $G[\bigcup_{i'\in [i,j]} U_{i'}]$ as a star $G'$ on $\bigcup_{i'\in [i,j]} U_{i'}$ with $u_i$ as the center vertex and $X=\{u_i,\dots, u_j\}$).
  Using $A$, we now apply dynamic programming. 
  To this end, we introduce a table $T[i,t]$ for $i\in [n]$ and $t\in [k]$. Entry~$T[i,t]$ is set to true if it is possible to partition the vertices from $\bigcup_{j\in [i]} U_j$ into $t$ connected districts that are $\ell$-fair and respect the size constraints. 
  We initialize the table by setting $T[i,1]$ to true if  $\bigcup_{j\in [i]} U_j$ is $\ell$-fair.
  Subsequently, for $t>1$, we update the table using:
  $$T[i,t]=\bigvee_{\substack{j\in [i-1] \wedge t',t''\in [t]: \\ t'+t''=t}} T[j,t'] \wedge A[j+1,i,t''].$$ 
	
  The reasoning behind this equality is that in a partitioning of $\bigcup_{j\in [i]} U_j$ into $t$ districts, there needs to be some $j\in [i-1]$ such that there is one district containing vertices $\{u_{j+1},\dots, u_i\}$. 
  However, not all vertices from $\bigcup_{t\in [j+1,i]} U_t$ need to be part of this district. 
  In fact, because of the connectivity constraints, there is some $t''$ such that one district contains vertices  $\{u_{j+1},\dots, u_i\}$ and all but  $t''-1$ vertices from $\bigcup_{t\in [j+1,i]} U_t$  and $t''-1$ districts consist of a single vertex from $\bigcup_{t\in [j+1,i]} U_t$ (such a partitioning respecting fairness and size constraints exists if and only if $A[j+1,i,t'']$ is true). 
  The remaining $t-t''$ districts then consists of vertices from  $\bigcup_{t\in [j]} U_t$ (such a partitioning respecting fairness and size constraints exists if and only if $T[j,t']$ is true).
  Table $T$ has $k\cdot n$ entries each computable in $\mathcal{O}(k \cdot n)$ 
  time. This leads to an overall running time of $\mathcal{O}(k \cdot n^3)$.  
\end{proof}

Note that a graph~$G$ has pathwidth one ($\pw(G) = 1$) if and only if $G$ is a disjoint union of caterpillars. 
By \Cref{prop:disjointunion} and \Cref{th:cater}, it follows that \fairCD is polynomial-time solvable on all graphs with pathwidth one (we later show in \Cref{co:pw2} that \fairCD is NP-hard on pathwidth-two graphs). 

\begin{corollary} \label{co:pw1}
 \fairCD on a graph~$G$ with $\pw(G) = 1$ is polynomial-time solvable. 
\end{corollary}

\subsection{An XP-Algorithm for \texorpdfstring{$\mln$}{mln}} \label{sub:mln}
Now, we generalize the polynomial-time solvability on paths and cycles
to a larger graph class.
More precisely, we develop an XP-algorithm for the max leaf number ($\mln$).
Recall that the max leaf number of a connected graph $G$ is the maximum number of leaves over all spanning trees of~$G$. 
Notably, any path or cycle has $\mln = 2$.
In order to develop a polynomial-time algorithm for constant $\mln$, we use the notion of \emph{branches}.
(See \Cref{fig:branches} for an illustration.)

\begin{definition}
  A \emph{branch} in a graph is either a maximal path in which all inner vertices have degree two or a cycle  in which all but one vertex have degree two. 
\end{definition}

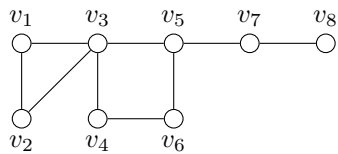
\begin{figure}
  \centering
  \begin{tikzpicture}[every node/.style={circle,draw,inner sep=0pt,minimum size=7pt}]
    \node at (0, 1) (v1) [label=above:$v_1$] {};
    \node at (0, 0) (v2) [label=below:$v_2$] {};
    \node at (1, 1) (v3) [label=above:$v_3$] {};
    \node at (1, 0) (v4) [label=below:$v_4$] {};
    \node at (2, 1) (v5) [label=above:$v_5$] {};
    \node at (2, 0) (v6) [label=below:$v_6$] {};
    \node at (3, 1) (v7) [label=above:$v_7$] {};
    \node at (4, 1) (v8) [label=above:$v_8$] {};

    \draw (v3) -- (v1) -- (v2) -- (v3);
    \draw (v3) -- (v4) -- (v6) -- (v5) -- (v3);
    \draw (v5) -- (v7) -- (v8); 
  \end{tikzpicture}
  \caption{An illustration of a graph with four branches: $(v_1, v_2, v_3), (v_3, v_5), \allowbreak (v_3, v_4, v_6, v_5), (v_5, v_7, v_8)$.}
  \label{fig:branches}
\end{figure}

Using a classical theorem of Kleitman and West~\cite{KW91}, 
Eppstein~\cite{Epp15} showed that any graph has at most $\mathcal{O}(\mln^2)$ 
branches.
For notational brevity, we assume that every branch $B$ that is a cycle has exactly one endpoint, namely, the vertex in $B$ with degree at least three in $G$ (if there is no such vertex in $B$, then we fix an arbitrary vertex as its endpoint).

\begin{theorem}\label{th:mln}
  \fairCD can be solved in $n^{\mathcal{O}(\mln^2)}$ time.
\end{theorem}
\begin{proof}
  Let $\mathcal{B}$ be the set of all branches of the given graph $G$ and let $X$ be the set of all endpoints of all branches.
  Suppose that there is a solution $\mathcal{V} = (V_1, \dots, V_k)$.
  Observe that there are naturally at most $ \vert X \vert $ subsets in~$\mathcal{V}$ containing at least one vertex from $X$.
  Without loss of generality, assume that $V_i$ contains at least one vertex from $X$ for $i \in [k']$ for some $k'\in [\min(k, \vert X \vert )]$, and let $X_i = X \cap V_i$.
  Let us now consider the relationship between a district~$V_i$ for $i\in [k']$ and a branch $B = (v_1, \dots, v_l) \in \mathcal{B}$ (where for all $i\in [l-1]$, $v_i$ and $v_{i+1}$ are adjacent in $G$ and if $B$ is a cycle, then $v_1 = v_l$).
  Since $V_i$ induces a connected subgraph in $G$, the following holds:
  \begin{itemize}
    \item
      If $v_1, v_l \notin X_i$, then $V_i$ and $B$ are disjoint:
      Since $V_i$ is connected and contains at least one vertex from $X \setminus \{v_1 ,v_l \}$, $V_i$ contains no ``inner'' vertices from $B$.
    \item
      If $v_1 \in X_i$ and $v_l \notin X_i$, then there is an integer $j \in [l - 1]$ such that $V_i \cap \{ v_1, \dots, v_l \} = \{ v_1, \dots, v_j \}$ (for $v_1 \notin X_i$ and $v_l \in X_i$, the situation is symmetric).
    \item
      If $v_1, v_l \in X_i$ for some $i\in [k]$, then there are two integers $j, j' \in [l]$ such that $V_i \cap \{ v_1, \dots, v_l \} = \{ v_1, \dots, v_j \} \cup \{v_{l - j' + 1}, \dots, v_l \}$.
      Note that $V_i$ contains all vertices of $B$ when $j + j' = l$. 
  \end{itemize}
	For all other districts~$V_i$, $i\in [k'+1,k]$, it holds that there is a branch $(v_1, \dots, v_l) \in \mathcal{B}$ with $V_i\subseteq \{v_2,\dots , v_{l-1}\}$. 

  Using these observations, our algorithm proceeds as follows.
  We iterate over all possible combinations of the following (note that the number of combinations is $n^{\mathcal{O}( \vert \mathcal{B} \vert )} = n^{\mathcal{O}(\mln^2)}$):
  \begin{itemize}
    \item
      An integer $k' \in [\min(k, \vert X \vert )]$.
    \item
      A partition $\mathcal{X} = (X_1, \dots, X_{k'})$ of $X$ into $k'$ subsets.
    \item
      For every branch $B = (v_1, \dots, v_l)\in \mathcal{B}$, two integers $j_B, j_B' \in [l]$ with $j_B+j_B' \leq l$.
  \end{itemize}
  Using these guesses, we can exactly determine the districts intersecting $X$: For every $i \in [k']$ and every branch $B =(v_1,\dots,v_l)\in \mathcal{B}$, we define $V_{i, B}$ as follows:
  \begin{itemize}
    \item
      If $v_1, v_l \notin X_i$, then $V_{i, B} := \emptyset$.
    \item
      If $v_1 \in X_i$ and $v_{l} \notin X_i$, then $V_{i, B}:=\{v_1,\dots, v_{j_B}\}$. 
      Symmetrically, if $v_1 \notin X_i$ and $v_{l} \in X_i$, then $V_{i,B}:=\{v_{l-j'_B+1},\dots, v_{l}\}$.
    \item
      If $v_1, v_l \in X_i$, then $V_{i, B}=\{v_1,\dots, v_{j_B}\}\cup \{v_{l-j'_B+1},\dots, v_{l}\}$.
  \end{itemize}
  For $i\in [k']$, let $V_i := \bigcup_{B \in \mathcal{B}} V_{i, B}$.
  We check whether the set $V_i$ is $\ell$-fair and respects the size constraints for every $i \in [k']$.
  If there is a $V_i$ that is not $\ell$-fair or violates the size constraints, then we proceed to the next combination.
  Otherwise, it remains to determine whether the vertices $V \setminus \bigcup_{i \in [k']} V_i$ can be partitioned into $k - k'$ $\ell$-fair districts that respect the size constraints.
  Since $G[V \setminus \bigcup_{i \in [k']} V_i]$ is a disjoint union of paths 
  as all endpoints of branches are contained in $\bigcup_{i \in [k']} V_i,$ this can be 
  done in polynomial time by \Cref{prop:disjointunion,prop:path}.
\end{proof}

We leave it open whether \fairCD parameterized by $\mln$ is fixed-parameter tractable or W[1]-hard.

\section{\textsc{FCD} on Trees and Tree-like Graphs} \label{se:trees}
After having seen in \Cref{se:paths} that \fairCD is polynomial-time solvable on paths, cycles, stars, and caterpillars, we now turn to trees.
In \Cref{sub:wtrees}, we prove that these polynomial-time results do not extend to trees. 
In particular, we prove that \fairCD on trees is NP-hard and W[1]-hard parameterized by $ \vert C \vert +k$. 
We complement these hardness results with an XP-algorithm for \fairCD on trees parameterized by $ \vert C \vert $ and another XP-algorithm for parameter~$k$. 
In fact, both XP-algorithms further extend to tree-like graphs: 
In \Cref{se:tw}, we prove that \fairCD parameterized by the treewidth of the given graph plus $ \vert C \vert $ is in XP. 
Subsequently, in \Cref{se:fen,se:fvn}, we consider the feedback edge number ($\fen$) and the feedback vertex number ($\fvn$), which are both alternative measures for the tree-likeness of a graph (the treewidth of a graph can be upper bounded in a function of $\fen$ and in a function of $\fvn$). Thus, the XP algorithm for $\twidth+~ \vert C \vert $ extends to the parameter combinations $\fvn+~ \vert C \vert $ and $\fen+~ \vert C \vert $, which is why we focus on the parameter combinations $\fvn+~k$ and $\fen+~k$ here. 
We prove that \fairCD parameterized by $\fen+~k$ is in XP and we show that there is (presumably) no such result for the treewidth and the feedback vertex number ($\fvn$) by showing that \fairCD is NP-hard even if $\fvn=1$, $\twidth=2$, and $k=2$. 

\subsection{W[1]-Hardness on Trees} \label{sub:wtrees}
In this subsection, we show that in contrast to paths, cycles, stars, and caterpillars,  
\fairCD on trees is NP-hard even without size constraints. Simultaneously, we show that \fairCD parameterized by the number $k$ of districts and the number $ \vert C \vert $ of colors is W[1]-hard on trees.
To this end, we present a parameterized reduction from the following version of 
\textsc{Grid Tiling} which is NP-hard and W[1]-hard  with respect to $t$ \cite{Mar07} (in 
this subsection, all indices are taken modulo~$t$): 

\decprob{Grid Tiling}{A collection $\mathcal{S}$ of $t^2$ (tile) sets $S^{i,j}\subseteq [m]\times [m]$, $i,j\in [t]$, where each $S^{i,j}$ consists of $n$~pairs and the first entries of all pairs (tiles) from $S^{i,j}$ sum up to the same number~$X$ for all $i,j\in [k]$ and the second entries of all pairs from $S^{i,j}$ sum up to the same number~$Y$ for all $i,j\in [k]$¸.}
{Can one choose one tile $(x^{i,j},y^{i,j})\in S^{i,j}$ for each $i,j\in [t]$ such that $x^{i,j}=x^{i+1,j}$ and $y^{i,j}=y^{i,j+1}$?}
Without loss of generality, we assume that $n>2$.

The general idea of the reduction is as follows.
Each solution to the constructed \fairCD instance has a \emph{center district}. 
The center district contains some large number~$Z$ of vertices of two ``dummy'' colors $c$ and $c'$.
The instance is constructed such that in every solution, vertices of any fixed color can appear at most $Z$ times in the center district. 
By definition, each tile $(x, y)$ belongs to one of $t^2$ tile sets $S^{i, j} \in \mathcal{S}$.
We construct a star $T_{x, y}^{i, j}$ for each tile such that for each tile set $S^{i, j} \in \mathcal{S}$, all but exactly one star $T_{x, y}^{i, j}$ need to be contained in the center district.
Thus, the center district (respectively its complement) basically encodes a selection of one tile from each tile set. 
Moreover, we construct the \fairCD instance in such a way that for two stars from two ``adjacent'' tile sets the respective first or second entries of the tiles need to match; otherwise the number of vertices of some color in the central district will exceed $Z$.  

\paragraph{Construction.} 
Let $(\mathcal{S},t,m,n,X,Y)$ be an instance of $\textsc{Grid Tiling}$. We construct an instance of \fairCD as follows. 

First of all, we set $\ell=0$, $\smin = 1$, $\smax = \infty$, and $k=t^2+1$.
For each $i,j\in [t]$, we introduce three distinct colors $b_{i,j}$, $d_{i,j}$, and $c_{i,j}$.
Moreover, we introduce three distinct colors $c$, $c'$, and $c^{\star}$. 
Now, we fix some constants which we use later. 
Let $W:=5n(t^2+t)+1$, $Z:=2(n-1)\cdot 5m W$, $f(i,j):=i\cdot t+j$, and $g(i,j):=t^2+t+i\cdot t+j$.
Note that this implies that $W\geq 2.5\cdot n \cdot g(i,j)$ and $W\geq 2.5\cdot n \cdot f(i,j)$ for all $i,j\in [t]$. 

We are now ready to construct the vertex-colored graph $G=(V,E)$.
We start by introducing a \emph{center vertex} $v_{\text{center}}$ of color $c^{\star}$. 
In a solution, we call the district containing $v_{\text{center}}$ the \emph{center district}. 
We add $Z$ vertices of color $c$ and $Z$ vertices of color $c'$, all of which are only adjacent to $v_{\text{center}}$. 

For each $i,j\in [t]$ and $(x,y)\in S^{i,j}$, we construct a star $T^{i, j}_{x, y}$ and connect its center to $v_{\text{center}}$.
We color the center of $T^{i, j}_{x, y}$ in $c^{\star}$.
Moreover, $T^{i, j}_{x, y}$ has the following leaves:
\begin{itemize}
 \item $\frac{Z}{2(n-1)}+W\cdot x -f(i,j)$ vertices of color $d_{i,j}$,
\item  $\frac{Z}{2(n-1)}-W\cdot x -f(i',j)$ vertices of color $d_{i',j}$ for $i' := i+1 \bmod k$,
\item $\frac{Z}{2(n-1)}+W\cdot y -g(i,j)$ vertices of color $b_{i,j}$,
\item $\frac{Z}{2(n-1)}-W\cdot y -g(i,j')$ vertices of color $b_{i,j'}$ for $j' := j + 1 \bmod k$, and
\item  $\max(\frac{Z}{2(n-1)}+W\cdot x-f(i,j),\frac{Z}{2(n-1)}+W\cdot y-g(i,j))$ vertices of color $c_{i,j}$.
\end{itemize}
Observe that the constructed star $T_{x,y}^{i,j}$ is $0$-fair, as the number of occurrences of $c_{i,j}$ matches the number of occurrences of the otherwise most frequent color. 
This concludes the construction.
\paragraph{Proof of Correctness.}  We start with a simple observation on the total number of vertices of some of the colors:
\begin{observation}\label{o:2} \label{o:1}
 For each $i,j\in [t]$, it holds that $\cv_{d_{i,j}}(V)= \frac{n}{n - 1} Z-2nf(i,j)$ and $\cv_{b_{i,j}}(V)= \frac{n}{n - 1} Z-2ng(i,j)$.
 For all $i,j\in [k]$, it holds that $\cv_{c_{i,j}}(V)\leq Z$.
\end{observation}
\begin{proof}
For each $i,j\in [t]$, vertices of color $d_{i,j}$ occur in $2n$ different stars, that is, in all stars corresponding to tiles from $S^{i,j}$ and $S^{i-1,j}$. 
As the first entries of all tiles from $S^{i,j}$ and of all tiles from $S^{i-1,j}$ sum up to $X$, it follows that 
\begin{align*}
  \cv_{d_{i,j}}(V)
  &= \sum_{(x, y) \in S^{i, j} } \frac{Z}{2(n - 1)} + W x - f(i, j) + \sum_{(x, y) \in S^{i - 1, j}} \frac{Z}{2(n - 1)} - Wx - f(i, j) \\
  &= \left(\frac{n}{2(n - 1)} Z + WX - n f(i, j)\right) + \left(\frac{n}{2(n - 1)} Z - WX - n f(i, j)\right) \\
  &= \frac{n}{n - 1} Z-2nf(i,j). 
\end{align*}
The same reasoning applies for all colors $b_{i,j}$ proving that $\cv_{b_{i,j}}(V)= \frac{n}{n - 1} Z-2ng(i,j)$. 
Lastly, for some $i,j\in [t]$,  vertices of color $c_{i,j}$ appear only in stars corresponding to tiles from $S^{i,j}$, each of them containing at most $\frac{Z}{2(n-1)}+W\cdot m$ such vertices. Thus, the number of vertices of color~$c_{i,j}$ is upper-bounded by $n\cdot(\frac{Z}{2(n-1)}+W\cdot m)$. 
As $W\cdot m= \frac{1}{8(n-1)}Z$ and $n>2$, this is smaller than $Z$ from which 
$\cv_{c_{i,j}}(V)\leq Z$ follows.
\end{proof}

Using \Cref{o:2}, we now prove the forward direction of the correctness of the construction.
\begin{lemma} \label{le:t-fo}
 If the given \textsc{Grid Tiling} instance is a yes-instance, then the constructed \fairCD instance is a yes-instance.
\end{lemma}
\begin{proof}
 Let $S=\{(x^{i,j},y^{i,j})\in S^{i,j} \mid i,j \in [t]\}$ be a solution for the given \textsc{Grid Tiling} instance. 
From this we construct a solution of the \fairCD instance as follows. 
For each $(x^{i,j},y^{i,j})\in S$, we create a separate district and put into this district all vertices from the star $T_{x,y}^{i,j}$ corresponding to $(x^{i,j},y^{i,j})$. We put all other vertices in the center district. 
Note that this construction respects the given number $k=t^2+1$ of districts and that all districts are by construction non-empty and connected. 

It remains to argue that all created districts are $0$-fair. 
Since every star is $0$-fair by construction, all non-center districts are $0$-fair.
For the center district, note that it contains $Z$ vertices of color $c$ and $Z$ vertices of color $c'$. 
Thus, it is sufficient to argue that there is no color such that the number of its occurrences in the center district exceeds~$Z$. 
For color~$c_{i,j}$ this directly follows from \Cref{o:1}. 
Next we consider color~$d_{i,j}$ for some $i,j\in[t]$.
All stars with some vertices of color~$d_{i, j}$ are part of the center district, except for the one corresponding to $(x^{i,j},y^{i,j})$ and the one corresponding to $(x^{i-1,j},y^{i-1,j})$.
Since $S$ is a solution, we have $x^{i-1,j}=x^{i,j}$.
Consequently, exactly $\frac{Z}{n-1}-2f(i,j)$ vertices of color~$d_{i,j}$ are excluded from the center district. 
Now it follows from \Cref{o:2} that the number of vertices of color~$d_{i, j}$ in the center district is at most~$Z$. 
An analogous argument also holds for $b_{i,j}$ for $i,j\in [t]$, which 
concludes the proof.
\end{proof}

It remains to prove the correctness of the backward direction.
To do this, we first observe that for every star, all of its vertices have to belong to the same district. 
Subsequently, we prove that the center district can contain at most $Z$ vertices of each color. 
We then show that in order to respect this bound, for each tile set, exactly one star corresponding to a tile from this set must be excluded from the center district. 
Again exploiting the fact that every color has at most $Z$ occurrences in the center district, we show that those excluded tiles form indeed a solution to the given \textsc{Grid Tiling} instance.
We start by observing that the vertices from one star need to be part of the same district.
\begin{observation} \label{ob:1}
 For each tile $(x,y)\in \mathcal{S}$, all vertices from the corresponding star need to be part of the same district in a solution to the constructed \fairCD instance.
\end{observation}
\begin{proof}
	As we set $\ell=0$ in the constructed \fairCD instance, there cannot exist a district containing just a single vertex. 
Thus, all vertices from a star need to belong to the same district as the 
center of the star.
\end{proof}

Using this observation, we make some more involved arguments dealing with the possible number of vertices of a color in the center district in a solution.
In the following, let $V_{\text{center}}$ denote the center district in a solution to the constructed \fairCD instance. We make this observation in order to ensure that no two different colors appear the same number of times and in particular more than $Z$ times in the center district.
\begin{lemma} \label{le:3}
  For each $i, j \in [t]$, if $\cv_{d_{i, j}}(V_{\text{center}}) \ge Z$, then the number $a$ of stars $T^{i, j}_{x, y}$ with some vertex of color $d_{i,j}$ that are not part of the center district is at most two.
  In particular, it holds that $\cv_{d_{i, j}}(V_{\text{center}}) = (\frac{n}{n-1} Z- 2n f(i, j)) - a (\frac{Z}{2(n-1)} - f(i,j)) + W q$ for some $q \in [-2m, 2m]$.
\end{lemma}
\begin{proof}
  Let us consider some $i,j\in [t]$.
  By \Cref{ob:1}, for every star $T_{x, y}^{i, j}$, all the vertices from~$T_{x, y}^{i, j}$ are in $V_{\text{center}}$ or no vertex from $T_{x, y}^{i, j}$ is in $V_{\text{center}}$. 

  By construction, every star with some vertex of color $d_{i,j}$ contains $\frac{Z}{2(n-1)}+W\cdot x-f(i,j)$ vertices of this color for some $x \in [-m, m]$.
  So as there are $\frac{n}{n-1}Z-2nf(i,j)$ vertices of color $d_{i,j}$ (\Cref{o:1}), if $a \ge 3$, then the number of vertices of color $d_{i,j}$ in the center district is at most $Z-\frac{Z}{2(n-1)}+3mW$. By the definition of $Z$ this is smaller than $Z$.

  Thus, we obtain $a \le 2$.
  Since there are in total $\frac{n}{n-1} Z- 2n f(i, j)$ vertices of color 
  $d_{i, j}$ by \Cref{ob:1}, the lemma holds.
\end{proof}

Analogously, we can prove similar bounds for the number of vertices of color $b_{i,j}$ for some $i,j\in [t]$ in the center district.
\begin{lemma} \label{le:4}
  For each $i, j \in [t]$, if $\cv_{b_{i, j}}(V_{\text{center}}) \ge Z$, then the number $a$ of stars $T^{i, j}_{x, y}$ with some vertex of color $b_{i,j}$ that are not part of the center district is at most two.
  In particular, it holds that $\cv_{b_{i, j}}(V_{\text{center}}) = (\frac{n}{n-1} Z - 2n g(i, j)) - a (\frac{Z}{2(n-1)} - g(i,j)) + W v$ for some $v \in [-2m, 2m]$.
\end{lemma}

Using these two lemmas, we can now prove that there are at most $Z$ vertices of the same color in the center district in a solution to the constructed \fairCD instance. 
\begin{lemma} \label{co:1}
  For every color $q$, $\cv_q(V_{\text{center}}) \le Z$.
\end{lemma}
\begin{proof}
  We claim that for two distinct colors $q, q'$ with $\cv_q(V_\text{center}), \cv_{q'}(V_\text{center}) > Z$, it holds that $\cv_q(V_\text{center}) \ne \cv_{q'}(V_\text{center})$.
  The statement of the lemma directly follows from the claim, as $V_{\text{center}}$ needs to be $0$-fair.

  Assume for a contradiction that $\cv_q(V_\text{center}) = \cv_{q'}(V_\text{center}) > Z$ for some colors $q \ne q'$.
  Since there are at most $Z$ vertices of color $c,c',c^{\star},c_{i,j}$ for $i,j\in [t]$, $q$ and $q'$ are among the colors $b_{i, j}$ and $d_{i, j}$ for $i, j \in [t]$.
  By \Cref{le:3,le:4} and as it holds that $\cv_q(V_\text{center}) = \cv_{q'}(V_\text{center})$, from this it follow that $2n h(i, j) + a\big(\frac{Z}{2(n - 1)} - h(i, j)\big) + Wv = 2n h'(i', j') + a'\big(\frac{Z}{2(n - 1)} - h'(i', j')\big) + Wv'$, for some $a, a' \in \{ 0, 1, 2 \}$, $v, v' \in [-2m, 2m]$, $h(i, j) \in \{ f(i, j), g(i, j) \}$, and $h'(i', j') \in \{ f(i', j'), g(i', j') \}$.
  Rewriting yields that $(5m(a - a') + v - v') W = (2n - a') h'(i', j') - (2n - a) h(i, j)$.
  Since $W$ is sufficiently large, the absolute value of the right hand side does not exceed $W$.
  Thus, we have $5m(a - a') + v - v' = 0$ and $(2n - a) h(i, j) = (2n - a') h'(i', j')$.
  The first equation implies $a = a'$ since $v - v' \in [-4m, 4m]$.
  Cancelling out $2n - a > 0$ in the second equation, we obtain $h(i, j) = h'(i', j')$.
  Now we have $q = q'$, as $f(\tilde{i},\tilde{j})\neq g(\tilde{i}',\tilde{j}')$ for all $\tilde{i},\tilde{j},\tilde{i}',\tilde{j}'\in [t]$, $f(\tilde{i},\tilde{j})= f(\tilde{i}',\tilde{j}')$ only if $\tilde{i}=\tilde{i}'$ and $\tilde{j}=\tilde{j}'$, and $g(\tilde{i},\tilde{j})= g(\tilde{i}',\tilde{j}')$ only if $\tilde{i}=\tilde{i}'$ and $\tilde{j}=\tilde{j}'$. 
  We have reached a contradiction.
\end{proof}

Using \Cref{co:1}, we can prove that each solution to the constructed \fairCD instance induces a selection of one tile from each tile set in the given \textsc{Grid Tiling} instance: 
\begin{lemma}\label{le:6}
 For each $i,j\in [t]$, exactly one star $T_{x, y}^{i, j}$  for some $(x,y)\in S^{i,j}$ is not part of $V_{\text{center}}$.
\end{lemma}
\begin{proof}
 \Cref{le:3} and \Cref{co:1} imply that for each $i,j\in [t]$ at least two stars containing vertices of color $d_{i,j}$ are not part of $V_{\text{center}}$. 
As each star $T_{x, y}^{i, j}$ contains vertices of colors $d_{i,j}$ and $d_{i+1,j}$ and as there exist $t^2$ non-center districts in the end, this means that for each $i,j\in [t]$ exactly two stars containing vertices of this color are not part of $V_{\text{center}}$. 

For the sake of contradiction, let us assume that there exists $i,j\in [t]$ such that two stars corresponding to tiles $(x, y), (x', y') \in S^{i,j}$ are not part of $V_{\text{center}}$. 
This implies by our previous observation that all other stars containing vertices of color $d_{i+1,j}$ need to belong to $V_{\text{center}}$. 
However, in this case, $\frac{Z}{n-1} - W(x + x') - 2 f(i, j) \le \frac{W}{n - 1} - 2W$ vertices of color $d_{i+1,j}$ are not part of $V_{\text{center}}$. 
By \Cref{o:2}, this implies that the number of vertices of color $d_{i+1,j}$ in $V_{\text{center}}$ is at least $Z+2W-2nf(i,j)$. 
As $W>n(t^2+t)$, this number is greater than~$Z$, contradicting \Cref{co:1}. 
\end{proof}

Putting all peaces together, we are now ready to prove the correctness of the backward direction of the construction.

\begin{lemma}\label{le:t-ba}
 If the constructed \fairCD instance is a yes-instance, then the given \textsc{Grid Tiling} instance is a yes-instance.
\end{lemma}
\begin{proof}
Let $\mathcal{V}$ be a solution to the constructed \fairCD instance and let $V_{\text{center}}\in \mathcal{V}$ be the district containing $v_{\text{center}}$. 
\Cref{le:6} implies that for each $i,j\in [t]$ exactly one star corresponding to a tile from $S^{i,j}$ is not part of $V_{\text{center}}$.
Let $$S:=\{(x^{i,j},y^{i,j})\in S^{i,j} \mid i,j \in [t]\wedge T^{i,j}_{x,y}\text{ is not part of $V_{\text{center}}$} \}$$ be a set of all tiles corresponding to excluded stars (one for each $i,j\in [t]$). 
We claim that $S$ is a valid solution to the given \textsc{Grid Tiling} instance.

For the sake of contradiction, assume that this is not the case. 
That is,  there either (a) exist  $i,j\in [t]$ such that $x^{i,j}\neq x^{i-1,j}$ or (b) exist $i',j'\in [t]$ such that $y^{i',j'}\neq y^{i',j'-1}$. 
Let us start by assuming that (a) is the case. 
Note that it is possible to assume without loss of generality that $x^{i,j}<x^{i-1,j}$, as from the fact that there exists some $x^{i,j}\neq x^{i-1,j}$ it follows that there also exists some $\tilde{i}\in [t]$ such that  $x^{\tilde{i},j}<x^{\tilde{i}-1,j}$.
The only stars containing vertices of color $d_{i,j}$ that are not part of $V_{\text{center}}$ are the star corresponding to $(x^{i,j},y^{i,j})$, which contains $\frac{Z}{2(n-1)}+W\cdot x^{i,j}-f(i,j)$ vertices of this color, and the star corresponding to $(x^{i-1,j},y^{i-1,j})$, which contains $\frac{Z}{2(n-1)}-W\cdot x^{i-1,j}-f(i,j)$ vertices of this color. 
This implies that at most $\frac{Z}{n-1}+W\cdot (x^{i,j}-x^{i-1,j})-2f(i,j)\leq \frac{Z}{n-1}-W-2f(i,j)$ vertices of this color are not part of $V_{\text{center}}$, where the inequality holds by our assumption that $x^{i,j}<x^{i-1,j}$. 
Combining this with \Cref{o:2}, it follows that the number of vertices of color
 $d_{i,j}$ in $V_{\text{center}}$ is at least $Z+W-2(n-1)f(i,j)$. By the definition of $W$, this is strictly  greater than $Z$. 
Applying \Cref{co:1}, we reach a contradiction. 

The same argument can also be applied if (b) holds, which proves that $S$ is a 
solution to the given \textsc{Grid Tiling} instance.
\end{proof}

From \Cref{le:t-fo} and \Cref{le:t-ba} the correctness of the reduction 
follows. As our construction takes only polynomial time and 
$ \vert C \vert +k=3t^2+4$ is bounded in a function of $t$, the NP-hardness and the
W[1]-hardness with respect to $ \vert C \vert +k$ of \fairCD on trees~follows.
\begin{theorem} \label{thm:whard}
\fairCD on trees is NP-hard and W[1]-hard with respect to $ \vert C \vert +k$, even if $\smin=1$ and $\smax=\infty$. 
\end{theorem}

Recall that by \Cref{co:pw1}, \fairCD can be solved in polynomial time on all graphs with pathwidth one (disjoint unions of caterpillars). 
Observe that the tree constructed in the reduction above has pathwidth two: Graph $G'$ obtained from $G$ by deleting $v_{\text{center}}$ is a disjoint union of stars.
Thus, $G'$ admits a path decomposition of width one.
Placing $v_{\text{center}}$ into every bag yields a path decomposition of $G$ of width two. This results in the following.
\begin{corollary}\label{co:pw2}
  \fairCD on graphs~$G$ with~$\pw(G) = 2$ is NP-hard and W[1]-hard with respect to~$ \vert C \vert +k$, even if $\smin=1$ and $\smax=\infty$. 
\end{corollary}
Notably, this result is tight in the sense that we have proved polynomial-time solvability on pathwidth-one graphs in \Cref{co:pw1}. 

As the tree constructed in the previous reduction has diameter four (it consists of a center vertex and centers of stars attached to it), we conclude that \fairCD is computationally intractable even on trees with a small constant diameter: 
\begin{corollary} \label{co:dia4}
\fairCD is NP-hard and W[1]-hard with respect to $ \vert C \vert +k$ on trees with diameter four, even if $\smin=1$ and $\smax=\infty$. 
\end{corollary}
We will complement this result in \Cref{co:dia3} where we show that \fairCD is polynomial-time solvable on trees with diameter at most three.

\subsection{An XP-algorithm for \texorpdfstring{$\mathbf{tw +  \vert C \vert }$}{tw +  \vert C \vert }} \label{se:tw}
Motivated by the hardness result from the previous subsection, we search for an XP-algorithm for the parameters $ \vert C \vert $ and $k$ for \fairCD on trees. We start by considering the parameter $ \vert C \vert $ here and the parameter $k$ in the next two subsections. Specifically, we show that there exists an XP-algorithm for $ \vert C \vert $ on tree-like graphs---more precisely, we present an XP-algorithm with respect to $ \vert C \vert +\twidth$, where $\twidth$ is the treewidth of the underlying graph.

Given a graph $G=(V,E)$, a tree decomposition  $(T = (V_T, E_T), \{B_x \}_{x \in V_T})$ of $G$ (as defined in the Preliminaries) is \emph{nice} if each node $x\in V_T$ has one of the following types\footnote{Note that in the following we refer to the elements of $V$ as vertices and the elements of $V_T$ as nodes. We refer to the set $B_x$ for a node $x \in V_T$ as a bag.}:
\begin{description}
  \item[Leaf node.] A leaf of $T$ with $B_x = \{ v \}$ for $v \in V$.
  \item[Introduce vertex $v$ node.] An internal node of $T$ with one child $y\in V_T$ such that $B_x = B_y \cup \{ v \}$.
  \item[Introduce edge $\{ u, v \}$ node.] An internal node of $T$ with one child $y\in V_T$ such that $u, v \in B_x = B_y$.
  \item[Forget $v$ node.] An internal node of $T$ with one child $y\in V_T$ such that $B_x = B_y \setminus \{ v \}$ for $v \in B_y$.
  \item[Join node.] An internal node of $T$ with two children $y\in V_T$ and $z\in V_T$ such that $B_x = B_y = B_z$.
\end{description} 
We will implicitly assume that every introduce edge $\{ u, v \}$ node is labeled by $\{ u, v \}$ and that for each edge there is exactly one such node.
Given a tree decomposition, a nice tree decomposition of equal width can be computed in linear time \cite{DBLP:books/sp/Kloks94}. 
By applying dynamic programming on top of the nice tree decomposition of the given graph, we establish the following: 

\begin{theorem}
\label{th:tw-XP}
\fairCD can be solved in $\mathcal{O}(n^{\mathcal{O}(\twidth \cdot  \vert C \vert )})$ time. 
\end{theorem}
\begin{proof}
  Suppose that we are given an instance $\mathcal{I}=(G,C,\col,k,\ell,\smin,\smax)$ of \fairCD admitting a solution $\mathcal{V} = \{ V_1, \dots, V_k \}$.
  Let $(T = (V_T, E_T), \{ B_x \}_{x \in V_T})$ be a tree decomposition of $G$. 
  For $x \in V_T$, let $G_x$ be the graph whose vertices and edges are those introduced in a node from the subtree of~$T$ rooted at~$x$. 
  Further, let $U_x$ be the set of vertices in~$G_x$.
  Our algorithm employs a dynamic programming method that traverses~$T$ from bottom to~top.

  Observe that each district $V_i \in \mathcal{V}$ from the solution is of one of the following three types with respect to each $x\in V_T$: 
  \begin{itemize}
        \item[(i)] $V_i \cap B_x \ne \emptyset$, i.e., the district $V_i$ overlaps with the vertices in the bag of $x$. 
        \item[(ii)] $V_i \cap B_x = \emptyset$ and $V_i \subseteq U_x$, i.e., the district $V_i$ is disjoint from $B_x$ but fully contained in vertices from bags from the subtree rooted at $x$. 
        \item[(iii)] $V_i \cap B_x = \emptyset$ and $V_i \subseteq V \setminus U_x$, i.e., the district $V_i$ is disjoint from $B_x$ and does not contain any vertex occurring in a bag from the subtree rooted at $x$. 
       \end{itemize}
  
  These three cases are exhaustive, as by definition of the tree decomposition, for each vertex, all nodes in which bags the vertex appears form a connected subgraph in $T$. 
  Note that for each district $V_i\in \mathcal{V}$ of type (i), the induced subgraph~$G_x[V_i]$ can have multiple connected components, in which case $G[V_i]$ includes some vertices or edges not in $G_x$.
  
  For each $x\in V_T$, we keep track of the following information:
  \begin{enumerate}
    \item[(1)] the intersection of $B_x$ and every $V_i\in \mathcal{V}$ of type (i),
    \item[(2)] the intersection of $B_x$ with the connected components of the subgraph $G_x[V_i]$ for all districts $V_i\in \mathcal{V}$ of type (i),
    \item[(3)] the number of occurrences of color $c$ for every color $c \in C$ in $V_i\cap U_x$ for  every district $V_i\in \mathcal{V}$ of type (i), and
    \item[(4)] the number of districts of type (ii).
  \end{enumerate}
	
  We capture this information in the following variables: For each $x\in V_T$, let (1)~$\pb$ be a partition of~$B_x$,  (2)~$\ppb$ be a partition of~$B_x$ where each subset from $\ppb$ is contained in a subset of $\pb$,  (3)~$\cc \colon \pb \to \mathbb{N}^{ \vert C \vert }$ be some function, and (4) $\kp \in \mathbb{N}$ be some integer with $\kp\leq k$. 
  
  On an intuitive level, moving in the tree from the bottom to the top, we keep track of the number of ``completed'' districts (type (ii) districts). 
  Moreover, for each ``unfinished'' district (type (i) district), we store its current color distribution (and thereby implicitly its current size of so far added vertices), and how the vertices from the current bag are distributed to the unfinished districts. 
  Also, for each unfinished district, we keep track of which of its vertices from the current bag are already connected over vertices  that have been added to the district and edges that have been introduced in the respective subtree.
  It is necessary to keep track of the connectivity of vertices from one district (stored in $\ppb$), as we need this information to ensure that we make the district connected by adding further vertices and edges to it further up in the tree. 
  Moreover, we cannot directly infer it from $\pb$, as we cannot remember all vertices put in the districts so far. 
  
  It remains to formally describe the dynamic program we use. 
 For all combinations of the variables $\pb$, $\ppb$, $\cc$, and $\kp$, we compute $A_x(\pb, \ppb, \cc, \kp) \in \{ 0, 1 \}$ such that $A_x(\pb, \ppb, \cc, \kp) = 1$ if and only if there is a partition of the vertices $U_x$ into exactly $ \vert \ppb \vert  + \kp$ subsets $(W_1,\dots, W_{ \vert \ppb \vert  + \kp})$ with the following properties:
  \begin{itemize}
    \item For every $i\in [ \vert \ppb \vert  + \kp]$, $G_x[W_i]$ is connected.
    \item Exactly $\kp$ subsets from $(W_1,\dots, W_{ \vert \ppb \vert  + \kp})$ are disjoint from $B_x$. Moreover, all of them are $\ell$-fair and respect the size constraints.
    \item For every subset $W_i$ intersecting $B_x$, it holds that $W_i \cap B_x \in \ppb$.
    \item For every $B \in \pb$, it holds that $\cc(B) = \sum_{i\in [ \vert \ppb \vert  + \kp], W_i \cap B \neq \emptyset} \cv(W_i)$.
  \end{itemize}
  From the subsets in the partitioning $(W_1,\dots, W_{ \vert \ppb \vert  + \kp})$, $\kp$ are $\ell$-fair connected districts of type (ii) with respect to $x$. 
  The other $ \vert \ppb \vert $ subsets are partly contained in $B_x$ and restricted to $B_x$ form the partition~$\ppb$ describing the intersection of $B_x$ with the connected components of $G_x[V_i]$ for districts~$V_i$ of type~(i). 
  Note that some of these $ \vert \ppb \vert $ subsets may end up in the same district in the 
  solution. 
   Moreover, assuming that $\pb$ corresponds to the intersections of $B_x$ and districts of type (i), it needs to hold that for each set in the intersection $B\in \mathcal{B}$,  $\cc(B)$ is the distribution of colors in the union of all subsets $W_i$ intersecting with $B$ (these are the subsets that will finally end up being one district in the solution).  
  
  Observe that a given instance of \fairCD is a yes-instance if and only if there exists a partition $\pb_r$ of $B_r$ for the root $r$ of $T$ and a function $\cc_r$ such that for all $B \in \pb_r$, $\cc_r(B)$ is $\ell$-fair and its $L_1$-norm is in the range $[\smin, \smax]$ and $A_r(\pb_r, \pb_r, \cc_r, k -  \vert \pb_r \vert ) = 1$.

  We compute $A_x(\pb, \ppb, \cc, \kp)$ by traversing the nodes $V_T$ of the tree $T$ from bottom to top. 
  Depending on the type of the current node $x\in V_T$, we apply one of the following five cases.
  In the following, let $y\in V_T$ (and $z\in V_T$) be the child node(s) of $x$ if $x$ has one (two) child node(s) in~$T$.
  \paragraph*{Leaf node $x$.}
  For a leaf node $x$ with $B_x = \{ v \}$, we have $A_x(\{ v \}, \{ v \}, \cc, 0) = 1$ if and only if $\cc(\{ v \}) = \cv(\{ v \})$.

  \paragraph*{Introduce vertex $v$ node $x$.}
  Note that $v$ is isolated in $G_x$.
  Thus, $A_x(\pb, \ppb, \cc, \kp) = 0$ if $\{ v \}$ is not contained in $\ppb$.
  Otherwise we have
  \[
    A_x(\pb, \ppb, \cc, \kp) =
      A_y(\pb', \ppb', \cc', \kp),
  \]
  where $\pb' = (\pb \setminus \{ B \}) \cup \{ B \setminus \{ v \} \}$ for $B \in \pb$ with $v \in B$ ($\pb' = \pb \setminus \{ B \}$ if $B = \{ v \}$), $\ppb' = \ppb \setminus \{ \{ v \} \}$, and $\cc'(B') = \cc(B) - \cv(\{ v \})$ if $B' = B \setminus \{ v \}$ (we assume that $\cc'(\emptyset)$ is a zero vector here) and $\cc'(B') = \cc(B')$ otherwise.

  \paragraph*{Introduce edge $\{u,v\}$ node $x$.} 
  If $u$ and $v$ belong to distinct subsets of $\mathcal{B}$, then we have $A_x(\pb, \ppb, \cc, \kp) = A_y(\pb, \ppb, \cc, \kp)$.
  We also have the same recurrence if $u$ and $v$ belong to the same subset of $\ppb$.
  This is because in both cases the connected components of districts of type (i) in $G_x$ do not change.
  Otherwise, we have
  \[
    A_x(\pb, \ppb, \cc, \kp) = A_y(\pb, \ppb, \cc, \kp) \vee \bigvee_{\ppb'} A_y(\pb, \ppb', \cc, \kp),
  \]
  where $\bigvee$ is over all partitions $\ppb'$ of $B_y$ from which $\ppb$ can be obtained by removing $D_u, D_v\in \ppb'$ with $u \in D_u$ and $v \in D_v$ from $\ppb'$ and adding $D_u \cup D_v$ to $\ppb'$.

  \paragraph*{Forget $v$ node $x$.}
  We have two cases.
  In the first case, the district $V_i\in \mathcal{V}$ to which $v$ belongs does not intersect $B_x$.
  This implies that $V_i$ is now fully contained in $G_x$ and becomes a district of type (ii). 
  Then $\{ v \} $ should have been part of $\pb$ and $\ppb$ for child $y$.
  Let us define
  \[
    A_x^1(\pb, \ppb, \cc, \kp) = 
    \bigvee_{\cc'} (A_y(\pb \cup \{ \{ v \} \}, \ppb \cup \{ \{ v \} \}, \cc', \kp - 1)),
  \] 
  where $\bigvee$ is over all functions $\cc'$ such that $\cc'(B) = \cc(B)$ for all $B \in \pb$, $\MOV(\cc'(\{ v \})) \le \ell$, and $\smin \le \sum_{c \in C}\cc_c'(\{ v \}) \le \smax$. 
  Note that we also need to decrease $\kp$, as $V_i$ is an additional district of type (ii) for the node $x$.
  In the other case, the district in which $v$ is contained intersects $B_x$.
  In this case, it holds that
  \[
    A_x^2(\pb, \ppb, \cc, \kp) =
    \bigvee_{\pb'\!,\,\ppb'} A_y(\pb', \ppb', \cc, \kp),
  \]
  where $\bigvee$ is over all partitions $\pb'$ and $\ppb'$ of $B_y$ such that $\pb' = (\pb \setminus \{ B \}) \cup (B \cup \{ v \})$ and $\ppb' = (\ppb \setminus \{ D \}) \cup (D \cup \{ v \})$ for some $B \in \pb$ and $D \in \ppb$.
  All in all, we have $A_x(\pb, \ppb, \cc, \kp) = A_x^1(\pb, \ppb, \cc, \kp) \vee A_x^2(\pb, \ppb, \cc, \kp)$.

  \paragraph*{Join node $x$.}
  Note that two vertices $u, w \in D$ for some $D \in \ppb$ can reach each other in $G_x$ via connections in $G_y$ and $G_z$.
  So we have
  \begin{align*}
  & A_x(\pb, \ppb, \cc, \kp) = \\
  & \bigvee_{\ppb_y, \ppb_z, \cc_y, \cc_z, \kp_y, \kp_z}  A_y(\pb, \ppb_y, \cc_y, \kp_y) \wedge A_z(\pb, \ppb_z, \cc_z, \kp_z).
  \end{align*}
  Here, $\bigvee$ is over all $\ppb_y, \ppb_z, \cc_y, \cc_z, \kp_y, \kp_z$ such that all the following hold:
  \begin{itemize}
    \item
      $\mathcal{D}$ are the connected components of the graph whose vertex set is $B_x$ and edge set is $\{ \{ u, w \} \mid (\exists D_y \in \mathcal{D}_y \colon u, w \in D_y) \vee (\exists D_z \in \mathcal{D}_z \colon u, w \in D_z) \}$.
    \item
      $\cc_y(B) + \cc_z(B) = \cc(B)-\cv(B)$ for every $B \in \mathcal{B}$.
    \item
      $\kp_y + \kp_z = \kp$.
  \end{itemize}

  Now we examine the running time.
  Note that there are $\twidth^{\mathcal{O}(\twidth)}$ partitions of $B_x$ for every node~$x$ of $T$.
  Since $\cc(B) \in \{ 0, \dots, n \}^{ \vert C \vert }$ for every $B \in B_x$ and $x \in V_T$, there are at most $n^{\mathcal{O}(\twidth \cdot  \vert C \vert )}$ states for~$\cc$.
  It is easy to see that the computation of all values takes 
  $n^{\mathcal{O}(\twidth \cdot  \vert C \vert )}$ time.
\end{proof}

Since a tree is of treewidth one, we have the following:

\begin{corollary}
 \fairCD on trees can be solved in $n^{\mathcal{O}( \vert C \vert )}$ time. 
\end{corollary}

\subsection{An XP-algorithm for \texorpdfstring{$\mathbf{fen + k}$}{fen + k}} \label{se:fen}
Having constructed a polynomial-time algorithm for constant $ \vert C \vert $ on trees and on graphs with a constant treewidth, we now consider the number $k$ of districts as our parameter for \fairCD on trees. 
We show that there is a simple XP-algorithm with respect to $k$ for \fairCD on trees, which naturally extends to an XP-algorithm with respect to $\fen+~k$, where $\fen$ is the number of edges that need to be deleted to make the given graph a~tree:

\begin{proposition}\label{th:fen}
  \fairCD can be solved in $n^{\mathcal{O}(\fen + k)}$ time.
\end{proposition}
\begin{proof}
  Suppose that there is a solution $(V_1, \dots, V_k)$.
  As for each $i\in [k]$, $V_i$ is connected (for which at least $ \vert V_i \vert -1$ edges are needed), 
  the number of edges whose endpoints both lie in the same district is at least $\sum_{i = 1}^k ( \vert V_i \vert  - 1) = n - k$.
  Since the input graph $G = (V, E)$ has at most $n + \fen - 1$ edges (by the definition of $\fen$), it follows that there are at most $(n + \fen - 1) - (n - k) = \fen + k - 1$ edges whose endpoints belong to different districts.
  Accordingly, in our algorithm for each edge set $E' \subseteq E$ of size at most $\fen + k - 1$, we verify whether $(V, E \setminus E')$ has $k$ connected components each of which being $\ell$-fair and respecting the size constraints.
  We return yes if this is the case for some subset of edges and no otherwise.
  Overall, it takes $(n + \fen - 1)^{\fen + k - 1} \cdot n^{\mathcal{O}(1)} = 
  n^{\mathcal{O}(\fen + k)}$ time to do so.
\end{proof}

As the treewidth of a graph is upper-bounded in a function of its feedback edge number, we can conclude from \Cref{th:tw-XP} that 
\fairCD parameterized by $\fen+~ \vert C \vert $ is in XP.
However, \Cref{th:fen} raises the question whether there is an XP-algorithm for~$\twidth+~k$. 
We answer this question negatively in the next subsection. 

\subsection{NP-hardness for \texorpdfstring{$\mathbf{fvn = 1}$}{fvn = 1}  and \texorpdfstring{$\mathbf{k = 2}$}{k = 2}}\label{se:fvn}

Having already considered the parameters treewidth and feedback edge number, we now consider a third way to measure the distance from a tree, the number of vertices to delete to make it a tree (feedback vertex number $\fvn$). 
As $\twidth+1\leq \fvn$, from \Cref{th:tw-XP} it follows that there is an XP-algorithm for $\fvn+~ \vert C \vert $.
We now prove that \fairCD is NP-hard even for $\fvn=1$ and $k=2$, in contrast to the result from the previous subsection, where we gave an $n^{\mathcal O(\fen + k)}$-time algorithm.
Notably, this NP-hardness result also excludes the existence of an XP-algorithm for~$\twidth+~k$ unless P $=$ NP:
\begin{theorem}
\label{th:fvn}
\fairCD is NP-hard for $\operatorname{fvn} = 1$ and $k=2$, even if $\smin=1$ and $\smax=\infty$.
\end{theorem}
The rest of the section is devoted to proving the theorem. 
We reduce from the NP-hard \textsc{Not-All-Equal 3-Sat} problem \cite{DBLP:conf/stoc/Schaefer78} where the input is a set $X$ of $n$ boolean variables and a set~$Y$ of $m$~clauses over~$X$ such that each clause $y\in Y$ contains three different literals, and the question is whether there exists a truth assignment to the variables in~$X$ such that for each clause $y\in Y$ at least one literal is set to true and at least one literal is set to false. 
Notably, given an assignment fulfilling these constraints, the assignment that assigns all variables in $X$ the opposite truth value also fulfills the constraints.

The general idea of the construction is that we introduce one vertex for each literal and that the two districts in the solution correspond to two opposite truth assignments of the variables from $X$ that are both solutions of the \textsc{Not-All-Equal 3-Sat} instance. 

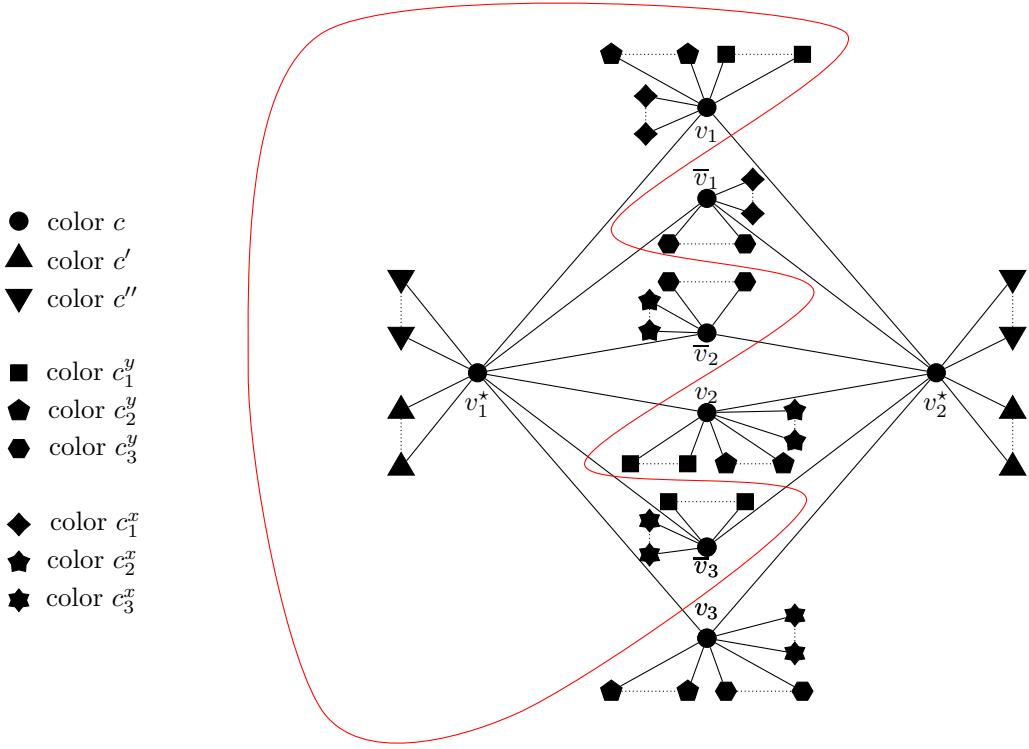
\begin{figure*}[bt]
	\begin{center}
	\resizebox{.9\textwidth}{!}{%
		\begin{tikzpicture}
		\node[vertex,label={[label distance=0.125cm]0:color $c$}] (1) at (-6, 2) {};
		\node[triangle,label={[label distance=0.125cm]0:color $c'$}] (2) at (-6, 1.5) {};
		\node[itriangle,label={[label distance=0.125cm]180:color $c''$}] (3) at (-6, 1) {};
		\node[square,label={[label distance=0.125cm]0:color $c^y_1$}] (4) at 
		(-6, 0) {};
		\node[penta,label={[label distance=0.125cm]0:color $c^y_2$}] (5) at 
		(-6, 
		-0.5) {};
		\node[hexa,label={[label distance=0.125cm]0:color $c^y_3$}] (5) at 
		(-6, -1) {};
		\node[star4,label={[label distance=0.125cm]0:color $c^x_1$}] (4) at 
		(-6, -2) {};
		\node[star5,label={[label distance=0.125cm]0:color $c^x_2$}] (5) at 
		(-6, -2.5) {};
		\node[star6,label={[label distance=0.125cm]0:color $c^x_3$}] (5) at 
		(-6, -3) {};
		
		\node[vertex,label=-90:$v^{\star}_1$] (vs1) at (0, 0) {};
		\node[vertex,label=-90:$v^{\star}_2$] (vs2) at (6, 0) {};
		\node[triangle] (v2) at (-1, -0.5) {};
		\node[triangle] (v3) at (-1, -1.25) {};
		
		\node[itriangle] (vv2) at (-1, 0.5) {};
		\node[itriangle] (vv3) at (-1, 1.25) {};
		\node[triangle] (2v2) at (7, -0.5) {};
		\node[triangle] (2v3) at (7, -1.25) {};
		
		\node[itriangle] (2vv2) at (7, 0.5) {};
		\node[itriangle] (2vv3) at (7, 1.25) {};
		
		\node[vertex,label={[label distance=-0.125cm]90:$\overline{v}_1$}] (w1) at (3, 2.3) {};
		\node[vertex,label=-90:$v_1$] (w2) at (3, 3.5) {};
		\node[vertex,label={[label distance=-0.1cm]-90:$\overline{v}_2$}] (w5) at (3, 0.525) {};
		\node[vertex,label={[label distance=-0.125cm]90:$v_2$}] (w6) at (3, -0.525) {};
		\node[vertex,label={[label distance=-0.125cm]-90:$\overline{v}_3$}] (w3) at (3, -2.3) {};
		\node[vertex,label=90:$v_3$] (w4) at (3, -3.5) {};
		
		\node[hexa] (1C1a) at (2.5, 1.7) {};
		\node[hexa] (1C1b) at (3.5, 1.7) {};
		\node[penta] (1C2a) at (1.75,4.2) {};
		\node[penta] (1C2b) at (2.75,4.2) {};
		\node[square] (1C3a) at (3.25,4.2) {};
		\node[square] (1C3b) at (4.25,4.2) {};
		
		\node[square] (3C1a) at (2.5, -1.7) {};
		\node[square] (3C1b) at (3.5, -1.7) {};
		\node[penta] (3C2a) at (1.75,-4.2) {};
		\node[penta] (3C2b) at (2.75,-4.2) {};
		\node[hexa] (3C3a) at (3.25,-4.2) {};
		\node[hexa] (3C3b) at (4.25,-4.2) {};
		
		\node[hexa] (2C1a) at (2.5, 1.2) {};
		\node[hexa] (2C1b) at (3.5, 1.2) {};
		\node[square] (2C2a) at (2,-1.2) {};
		\node[square] (2C2b) at (2.75,-1.2) {};
		\node[penta] (2C3a) at (3.25,-1.2) {};
		\node[penta] (2C3b) at (4,-1.2) {};
		
		\draw[] (vs1)--(v2);
		\draw[] (vs1)--(v3);
		\draw[densely dotted] (v2)--(v3);
		\draw[] (vs1)--(vv2);
		\draw[] (vs1)--(vv3);
		\draw[densely dotted] (vv2)--(vv3);
		\draw[] (vs2)--(2v2);
		\draw[] (vs2)--(2v3);
		\draw[densely dotted] (2v2)--(2v3);
		\draw[] (vs2)--(2vv2);
		\draw[] (vs2)--(2vv3);
		\draw[densely dotted] (2vv2)--(2vv3);
		
		\draw[] (vs1)--(w1);
		\draw[] (vs1)--(w2);
		\draw[] (vs1)--(w3);
		\draw[] (vs1)--(w4);
		\draw[] (vs1)--(w5);
		\draw[] (vs1)--(w6);
		\draw[] (vs2)--(w1);
		\draw[] (vs2)--(w2);
		\draw[] (vs2)--(w3);
		\draw[] (vs2)--(w4);
		\draw[] (vs2)--(w5);
		\draw[] (vs2)--(w6);
		
		\draw[] (w1)--(1C1a);
		\draw[] (w1)--(1C1b);
		\draw[densely dotted] (1C1a)--(1C1b);	
		\draw[] (w2)--(1C2a);
		\draw[] (w2)--(1C2b);
		\draw[densely dotted] (1C2a)--(1C2b);
		\draw[] (w2)--(1C3a);
		\draw[] (w2)--(1C3b);
		\draw[densely dotted] (1C3a)--(1C3b);
		
		\draw[] (w3)--(3C1a);
		\draw[] (w3)--(3C1b);
		\draw[densely dotted] (3C1a)--(3C1b);	
		\draw[] (w4)--(3C2a);
		\draw[] (w4)--(3C2b);
		\draw[densely dotted] (3C2a)--(3C2b);
		\draw[] (w4)--(3C3a);
		\draw[] (w4)--(3C3b);
		\draw[densely dotted] (3C3a)--(3C3b);
		
		\draw[] (w5)--(2C1a);
		\draw[] (w5)--(2C1b);
		\draw[densely dotted] (2C1a)--(2C1b);	
		\draw[] (w6)--(2C2a);
		\draw[] (w6)--(2C2b);
		\draw[densely dotted] (2C2a)--(2C2b);
		\draw[] (w6)--(2C3a);
		\draw[] (w6)--(2C3b);
		\draw[densely dotted] (2C3a)--(2C3b);
		\draw [red] plot [smooth cycle] coordinates {(-3,0) (-2,4.5) (4.8,4.5) 
		(1.75,1.9) (4.4,1.05) (1.4,-1.2) (4.3,-1.7) (0.5,-4.5) (-2,-4.5)};
		
		\node[star4] (v11) at (2.2, 3.65) {};
		\node[star4] (v12) at (2.2, 3.15) {};
		\draw[] (w2)--(v11);
		\draw[] (w2)--(v12);
		\draw[densely dotted] (v11)--(v12);	
		
		\node[star4] (v21) at (3.6, 2.1) {};
		\node[star4] (v22) at (3.6, 2.55) {};
		\draw[] (w1)--(v21);
		\draw[] (w1)--(v22);
		\draw[densely dotted] (v21)--(v22);	
		
		\node[star5] (v31) at (2.25, 0.95) {};
		\node[star5] (v32) at (2.25, 0.55) {};
		\draw[] (w5)--(v31);
		\draw[] (w5)--(v32);
		\draw[densely dotted] (v31)--(v32);	
		
		\node[star5] (v41) at (4.15, -0.9) {};
		\node[star5] (v42) at (4.15, -0.5) {};
		\draw[] (w6)--(v41);
		\draw[] (w6)--(v42);
		\draw[densely dotted] (v41)--(v42);
		
		\node[star6] (v51) at (4.15, -3.7) {};
		\node[star6] (v52) at (4.15, -3.2) {};
		\draw[] (w4)--(v51);
		\draw[] (w4)--(v52);
		\draw[densely dotted] (v51)--(v52);	
		
		\node[star6] (v61) at (2.25, -2.4) {};
		\node[star6] (v62) at (2.25, -1.95) {};
		\draw[] (w3)--(v61);
		\draw[] (w3)--(v62);
		\draw[densely dotted] (v61)--(v62);	
		
		\node[vertex,label={[label distance=-0.125cm]-90:$\overline{v}_3$}] (w3) at (3, -2.3) {};
		\node[vertex,label=90:$v_3$] (w4) at (3, -3.5) {};
	
		\end{tikzpicture}}
	\end{center}
	\caption{Example of the hardness reduction from \Cref{th:fvn} for the 
	\textsc{Not-All-Equal 3-Sat} instance~$X=\{x_1,x_2,x_3\}$, 
	$Y=\{y_1=\{x_1,x_2,\overline{x}_3\}, y_2=\{x_1,x_2,x_3\}, 
	y_3=\{\overline{x}_1,\overline{x}_2,x_3\} \}$. The district marked in red 
	corresponds to the solution setting $x_1$ to 
	true and $x_2$ and $x_3$ to false.}\label{fi:fvn}
\end{figure*}

\paragraph{Construction.} Given an instance $(X=\{x_1,\dots, x_n\},Y=\{y_1,\dots, y_m\})$ of \textsc{Not-All-Equal 3-Sat}, we construct an instance of \fairCD as follows. 
We add a color~$c_i^x$ for each variable $x_i\in X$ and a color~$c^y_j$ for each clause $y_j\in Y$. 
In addition, we add three colors $c$, $c'$, and~$c''$. 
Moreover, we set $k=2$, $\ell=0$, $\smin=1$, and $\smax=\infty$. 
Let $Z:=2\cdot n\cdot m+1$. 

We start the construction of the vertex colored graph $G=(V,E)$ by introducing two \emph{central vertices} $v^{\star}_1$ and $v^{\star}_2$ of color~$c$. 
We will construct the instance in a way that these two vertices need to lie in different districts. 
For each central vertex $v^{\star}_i$, $i \in \{ 1, 2 \}$, we introduce $3Z$~vertices of color~$c'$ and $3Z$~vertices of color~$c''$ and connect them to~$v^{\star}_i$.

Subsequently, for each variable $x_i\in X$, we introduce two \emph{literal vertices} $v_{i}$ and $\overline{v}_{i}$ of color $c$ and connect both these vertices to the two central vertices. 
For each literal vertex $\tilde{v}_i \in \{ v_i, \overline{v}_i \}$, we introduce $3Z-i$ vertices of color $c^x_{i}$ and connect them to $\tilde{v}_i$ 
(these vertices make sure that the two literal vertices end up in different districts).
For each clause $y_j\in Y$ in which $x_i$ occurs positively, we introduce $Z+j$ vertices of color $c^y_{j}$ and connect them to $v_{i}$.
For each clause $y_j\in Y$ in which $x_i$ occurs negatively, we introduce $Z+j$ vertices of color $c^y_{j}$ and connect them to $\overline{v}_{i}$ (these 
vertices ensure that there is no clause in which all three literal vertices corresponding to literals from the clause lie in the same district).
See \Cref{fi:fvn} for a visualization of the construction. 

We start by showing the forward direction of the correctness of the construction.  

\begin{lemma}\label{le:fvn1}
If the given {\normalfont\textsc{Not-All-Equal 3-Sat}} instance is a yes-instance, then the constructed {\fairCD} instance is a yes-instance.
\end{lemma}
\begin{proof}
Let $X'\subseteq X$ be the set of variables set to true in a solution to the given {\normalfont\textsc{Not-All-Equal 3-Sat}} instance. 
From this, we construct a solution $(V_1, V_2)$ to the constructed \fairCD instance as follows. 
We include $v^{\star}_1$ and all leaves attached to it in $V_1$. 
Moreover, we include in $V_1$ the following vertices: $v_i$ and all leaves attached to it for all $x_i\in X'$ and $\overline{v}_i$ and all leaves attached to it for all $x_i \in X \setminus X'$.
We include all other vertices in $V_2$.

It is easy to verify that $V_1$ and $V_2$ are both connected.
We will show that $V_1$ is $0$-fair.
By symmetry, an analogous argument will show that $V_2$ is $0$-fair.
First, observe that $V_1$ contains exactly $3Z$ vertices of color $c'$ and $3Z$ vertices of color $c''$.
We show that  $\cv_{\tilde{c}}(V_1)\leq 3Z$  for every color $\tilde{c} \in C$.
For color~$c$, we have that $\cv_{c}(V_1) \le \cv_{c}(V) = n+2 < 3Z$.
For each variable $x_i\in X$, as the two corresponding literal vertices are part of different districts, we have that $\cv_{c^x_i}(V_1)= 3Z-i < 3Z$.
For each clause $y_j\in Y$, as $X'$ is a solution, either one or two literal vertices corresponding to literals in $y_j$ and the attached leaves are part of $V_1$.
Thus, the number of vertices of color $c^y_j$ in $V_1$ is either $Z+j$ or $2Z+2j$. 
As $Z>2m$, it follows that $\cv_{c^y_j}(V_1) < 3Z$.
Thus, both districts $V_1$ and~$V_2$ are $0$-fair.
\end{proof}

It remains to prove the correctness of the backward direction of the reduction. For this, note that the \fairCD instance is constructed such that the two central vertices need to end up in different districts and for each $x_i \in X$, the two literal vertices $v_i$ and~$\overline{v}_{i}$ need to end up in different districts.
Thus, the two districts correspond to two inverse truth assignments. 
Subsequently, we will prove that there is no clause in which all corresponding vertices are in the same district. 
This will show that the two truth assignments induced by the two districts are a solution to the given \textsc{Not-All-Equal 3-Sat} instance. 

We start by observing that all leaves need to lie in the same district as the vertex they are attached to.
\begin{observation}\label{ob:fvn1}
  In a solution to the constructed \fairCD instance, all leaves attached to a vertex $v\in V$ need to lie in the same district as $v$.
\end{observation}
\begin{proof}
  If a leaf is not in the same district as its neighbor, then the leaf needs to form its own district.
  However, this is not possible, as we require each district to be $0$-fair. 
\end{proof}

Next, we show that indeed the two central vertices $v^{\star}_1$ and $v^{\star}_2$ always lie in different districts.
\begin{observation} \label{ob:fvn2}
The two central vertices $v^{\star}_1$ and $v^{\star}_2$ cannot belong to the same district in any solution to the constructed \fairCD instance. 
\end{observation}
\begin{proof}
Assume that $v^{\star}_1$ and $v^{\star}_2$ belong to the same district.
By \Cref{ob:fvn1}, the other district~$V'$ needs to consist of a single literal vertex $v_i$ or $\overline{v}_i$ for some $x_i\in X$ and all leaves attached to it. 
The most frequent color appearing in $V'$ is $c^x_i$ which occurs $3Z-i > 2Z$ times (by the definition of $Z$). 
The second most frequent color is either $c^y_j$ for some $j\in [m]$ occurring $Z+j < 2Z$ times or $c$ occurring once (by the definition of $Z$). 
Thus, the district cannot be $0$-fair.
\end{proof}

To prove that, for each variable, the two corresponding literal vertices need to be part of different districts and that not all literal vertices corresponding to literals from a clause can lie in the same district, we prove that neither district can contain more than $3Z$ vertices of the same color.
\begin{lemma}\label{le:fvn2a}
For any solution to the constructed \fairCD instance, no district contains more than $3Z$ vertices of the same color. 
\end{lemma}
\begin{proof}
Assume that there exists a solution with a district $V'$ containing more than $3Z$ vertices of the same color. 
Since $V'$ is $0$-fair, there are two colors $q$ and $q'$ with $\cv_q(V') = \cv_{q'}(V') > 3Z$.
We have $\cv_{c}(V') \le \cv(V) = n + 2 < 3Z$ and $\cv_{c'}(V') = \cv_{c''}(V') = 3Z$ by \Cref{ob:fvn2}.
Thus, we have $q, q' \in \{ c_{i}^x \mid x_i \in X \} \cup \{ c_{j}^y \mid y_j \in Y \}$.

For $x_i\in X$, there are two literal vertices each of which have $3Z-i$ leaves of color $c^x_i$ attached to it. 
For $y_j\in Y$, there are three literal vertices each of which have $Z+j$ leaves of color $c^y_j$ attached to it.
It follows from \Cref{ob:fvn1} that $\cv_{c^x_i}(V') \in \{ 0, 3Z - i, 6Z - 2i \}$ for each $x_i \in X$ and that $\cv_{c^y_j}(V') \in \{ 0, Z + j, 2Z + 2j, 3Z + 3j \}$ for each $y_j \in Y$.
Since $i \le n$, $j \le m$, $Z = 2\cdot n\cdot m+1$, and $\cv_q(V') = \cv_{q'}(V')$ for 
$q, q' \in \{ c_{i}^x \mid x_i \in X \} \cup \{ c_{j}^y \mid y_j \in Y \}$, it needs to hold that $\cv_q(V') = \cv_{q'}(V') = 0$, which contradicts $\cv_q(V') = 
\cv_{q'}(V') > 3Z$.
\end{proof}

Recalling that $3Z-i$ vertices of color~$c^x_i$ are attached to each literal vertex corresponding to $x_i\in X$ and that $6Z-2i>3Z$, the next observation directly follows from the previous lemma and \Cref{ob:fvn1}.
\begin{observation} \label{ob:fvn4}
Let $(V_1,V_2)$ be a solution of the constructed \fairCD instance. 
Then, for each $x_i\in X$, exactly one of the two corresponding literal vertices $v_i$ and $\overline{v}_i$ is part of $V_1$ and the other is part of~$V_2$.
\end{observation}

We are now ready to prove the backward direction of the correctness of our construction.
\begin{restatable}{lemma}{fvnn}\label{le:fvn2}
 If the constructed {\fairCD} instance is a yes-instance, then the given \textsc{Not-All-Equal 3-Sat} instance is a yes-instance.
\end{restatable}
\begin{proof}
Let $(V_1,V_2)$ be a solution to the constructed \fairCD instance. 
From  \Cref{ob:fvn4}, it follows that for each $x_i\in X$, exactly one of $v_i$ and $\overline{v}_i$ is part of $V_1$. 
Let $\varphi$ be the truth assignment induced by $V_1$, i.e., $\varphi$ sets $x_i$ to true if $v_i\in V_1$ and $x_i$ to false if $\overline{v}_i\in V_1$. 
We claim that $\varphi$ is a solution to the given  \textsc{Not-All-Equal 3-Sat} instance. 
Firstly, for the sake of contradiction, assume that there exists a clause $y_j\in C$ containing only literals that are satisfied by~$\varphi$. 
However, by \Cref{ob:fvn1}, this implies that all vertices of color $c^y_j$ are part of~$V_1$, contradicting \Cref{le:fvn2a} as there exist $3Z+3j$ such vertices. 
Secondly, for the sake of contradiction, assume that $y_j$ contains no literal satisfied by~$\varphi$. 
However, by \Cref{ob:fvn1},  this implies that all vertices of color $c^y_j$ are part of~$V_2$, contradicting again \Cref{le:fvn2a}. 
Consequently, $\varphi$ is a solution.
\end{proof}

Observing that $\{v^{\star}_1\}$ is a feedback vertex set of the constructed graph and that the construction can be computed in polynomial time, \Cref{th:fvn} follows directly from \Cref{le:fvn1} and \Cref{le:fvn2}

From our reduction, we can further conclude that \fairCD is also para-NP-hard with respect to the treewidth plus the number $k$ of districts.
\begin{corollary} \label{co:tw-hard}
\fairCD is NP-hard for $\twidth = 2$ and $k=2$, even if $\smin=1$ and $\smax=\infty$.
\end{corollary}

\section{\textsc{FCD} on Graphs of Bounded Vertex Cover Number} \label{se:vc}
Motivated by our hardness results for graphs with constant treewidth, we now turn to the size $\vc$ of a minimum vertex cover, a parameter never smaller than the treewidth. 
In this section, we present two parameterized algorithms, namely, an XP-algorithm for~$\vc$ and an FPT-algorithm for $\vc +  \vert C \vert $. 
Unfortunately, we were unable to settle whether \fairCD parameterized by~$\vc$ is W[1]-hard or fixed-parameter tractable. 
In contrast, we develop an FPT-algorithm for the number of vertices with degree at least two (a parameter which is never smaller than~$\vc$).

Both algorithms for $\vc$ rely on the following~lemma.

\begin{lemma}\label{obs:vc}
  Let $S$ be a vertex cover of minimum size and $\mathcal{V} = (V_1, \dots, V_k)$ a solution to an \fairCD instance on a graph~$G$. There are at most $\vc$ districts in $\mathcal{V}$ that contain at least one vertex from $S$.
  Moreover, for every $V_i\in \mathcal{V}$ with $S \cap V_i \ne \emptyset$, there is a set $J_i \subseteq V_i \setminus S$ of at most $ \vert S \cap V_i \vert -1$ vertices such that $G[(S \cap V_i) \cup J_i]$ is connected.
\end{lemma}
\begin{proof}
 As each vertex is only contained in one district, the first part of the lemma follows from the definition of $\vc$. 
 To prove the second part, fix some $V_i\in \mathcal{V}$.
 Consider a minimum spanning tree $T = (V_i, F)$ of $G[V_i]$.
 Let $J_i \subseteq V_i \setminus S$ be the set of vertices of degree at least two in $T$.
 Let $T' := ((S \cap V_i) \cup J_i, F')$ be the result of deleting from $T$ each vertex $v \in V_i \setminus (S \cup J_i)$ along with an edge incident to it. 
 Observe that $T'$ is connected and thus $G[(S \cap V_i) \cup J_i]$ is connected, which contains $T'$ as a subgraph.
 We show that $\vert J_i \vert \le \vert S \cap V_i \vert - 1$.
 As by the definition of $J_i$ vertices from $J_i$ are only adjacent to vertices from $S\cap V_i$ and as every vertex from $J_i$ has degree at least two in $T$ and $T'$, we have $\vert F' \vert \ge 2 \vert J_i \vert$.
 We also have $\vert F' \vert = \vert S \cap V_i \vert + \vert J_i \vert - 1$ since $T'$ is a tree.
 Thus, we have $\vert J_i \vert \le \vert S \cap V_i \vert - 1$.
\end{proof}

We first show that \fairCD parameterized by $\vc$ is in XP using the following approach:
We first guess how the vertex cover is partitioned into districts in the sought solution, which gives partial (not necessarily connected) districts.
For every partial district, we then guess some vertices outside the vertex cover to include such that the partial district becomes connected, the two most frequent colors, and how often they occur in the resulting district. 
The remaining problem can be reduced to the polynomial-time solvable \textsc{$(g, f)$-Factor} problem~\cite{DBLP:conf/stoc/Gabow83}. 
\begin{theorem} \label{thm:xpvc}
    \fairCD is solvable in $n^{\mathcal{O}(\vc)}$ time.
\end{theorem}
\begin{proof}
  Suppose there is a solution $\mathcal{V} = (V_1, \dots, V_k)$ to the given \fairCD instance.
  Let $S$ be a vertex cover of minimum size and let $I = V \setminus S$ be the vertices outside of the vertex cover.
  As in the proof of \Cref{thm:fptvcc}, our algorithm first tries all possibilities to fix some structure with respect to $S$.
  Then, we will construct an instance of the polynomial-time solvable \textsc{$(g, f)$-Factor} problem, which generalizes \textsc{Maximum Matching}, for each choice of the following:
  \begin{itemize}
    \item
      An integer $k' \le \min(k,\vc)$.
      If $\ell = 0$ or $\smin \ge 2$, then we consider only one choice for $k'$, namely $k' = k$.
    \item
      A partition of $S$ into $k'$ non-empty subsets $S_1, \dots, S_{k'}$.
      There are at most $\vc^{\vc}$ such partitions.
    \item
      For every $i \in [k']$, a set $J_i \subseteq I$ of at most $ \vert S_i \vert  - 1$ vertices such that $G[S_i \cup J_i]$ is connected and $J_i \cap J_{i'} = \emptyset$ for $i \ne i' \in [k']$.
      We can assume that $ \vert J_i \vert  \le  \vert S_i \vert  - 1$ vertices are sufficient to make $G[S_i]$ connected by \Cref{obs:vc}.
      The number of choices for $J_i$ is at most $\prod_{i \in [k']} n^{ \vert S_i \vert  - 1} \le n^{\vc}$.
    \item
      For every $i \in [k']$, a pair $(c_i, c_i')$ of colors, which have the largest numbers of occurrences in the sought $V_i$ among all colors.
      There are at most $ \vert C \vert ^{2\vc}$ choices for all pairs of colors.
    \item
      For every $i \in [k']$, the numbers $z_{c_i}, z_{c_i'}$ of occurrences of color $c_i$ and $c_i'$, respectively, with $z_{c_i} - \ell \leq z_{c_i'} \leq z_{c_i}$.
      Since $z_{c_i} \le n$ and $z_{c_i'} \le n$, there are at most $n^{2\vc}$ choices.
  \end{itemize}
  Since $ \vert C \vert  \le n$, there are at most $n^{\mathcal{O}(\vc)}$ choices. Let $I' = I \setminus \bigcup_{i \in [k']} J_i$.
  Now the question is whether we can partition the vertices $V$ into $k$ districts $(V_1, \dots, V_k)$ such that, for $i\in [k',k]$, $V_i$ consists of a single vertex from $I'$, and for $i\in [k']$, $S_i\cup J_i\subseteq V_i$, $G[V_i]$ is connected, $ \vert V_i \vert \in [\smin,\smax]$, $\cv_{c_i}(V_i)=z_{c_i}$, $\cv_{c_i'}(V_i)=z_{c_i'}$ and $\cv_{c}(V_i)\leq \cv_{c_i'}(V_i)$ for all $c\in C\setminus \{c_i\}$. 
 We reject the current combination, if for some $i\in [k']$ $\cv_{c_i}(S_i\cup J_i)> z_{c_i}$, $\cv_{c_i'}(S_i\cup J_i)>z_{c_i'}$, $\cv_{c}(S_i\cup J_i)> z_{c_i'}$ for some $c\in C\setminus \{c_i\}$, or $ \vert S_i \cup J_i \vert  > \smax$.
  Otherwise, to decide whether it is possible to distribute the vertices $I'$ to construct a partition respecting the above-described properties, we reduce to \textsc{$(g, f)$-Factor}, where given a graph $H=(U,F)$ and two functions $g,f:V\mapsto \mathbb{N}$, the question is whether there is a subgraph $H' = (U, F')$ with $F' \subseteq F$ such that every vertex has at least $g(v)$ and at most $f(v)$ neighbors in $H'$.
  
  We construct a bipartite graph $H$ and $f, g$ as follows.
  The left bipartition of $H$ consists of all vertices $v\in I'$.
  For every $i \in [k']$, we add the following vertices to the right bipartition:
  \begin{itemize}
    \item
      For each color $c \in \{ c_i, c_i' \}$, we add $z_c-\cv_{c}(S_i\cup J_i)$ vertices (call them $A_i$) and connect them via edges to all vertices $v\in I'$ having color $c$ in $G$ which have at least one neighbor in $S_i$ in $G$.
    \item
      For each color $c \in C \setminus \{ c_i, c_i' \}$, we add $z_{c_i'}-\cv_{c}(S_i\cup J_i)$ vertices (call them $B_i^c$) and connect them via edges to all vertices $v\in I'$ having color $c$ in $G$ which have at least one neighbor in $S_i$ in $G$.
      Let $B_i := \bigcup_{c \in C \setminus \{ c_i, c_i' \}} B_i^c$.
  \end{itemize}
  If $\smin = 1$, then we add a set $A'$ of $k - k'$ vertices to the right bipartition and connect them via edges to all vertices of $I'$.
  For every vertex $v$ introduce thus far, let $f(v) = g(v) = 1$.
  Finally, for every $i \in [k']$, add a vertex $v_i$ to the left side and let $f(v_i) =  \vert B_i \vert  - \max(0, \smin -  \vert A_i\cup S_i \cup J_i \vert )$ and $g(v_i) =  \vert B_i \vert  - \smax +  \vert A_i\cup S_i \cup J_i \vert $ (if it does not hold that $0\leq g(v_i)\leq f(v_i)$, then we continue with the next combination). We connect $v_i$ to all vertices from $B_i^c$.

  If the constructed \textsc{$(g, f)$-Factor} instance is a yes-instance, we return yes; otherwise, we continue with the next combination. 
  To see why the algorithm works correctly, suppose that there is a subgraph $H' = (U, F')$ of $H$ such that the degree of every vertex $v$ in $H'$ is in the range $[g(v), f(v)]$.
  From $H'$, we can construct a solution of the given \fairCD instance by 
  putting each vertex $v\in I'$ into the district corresponding to the neighbor of $v$ in $H'$ (if $\smin = 1$, then every vertex adjacent to some vertex from $A'$ forms a district of its own).
  Moreover, for all $i\in [k']$, we add $S_i\cup J_i$ 
  to $V_i$.
  As all vertices from~$I'$ have one neighbor in 
  $H'$, each vertex from $I'$ is assigned to a district. 
  As all vertices from $A$ have one neighbor in $H'$, for $i\in [k']$, it holds 
  that $\cv_{c_i}=z_{c_i}$ and $\cv_{c'_i}=z_{c'_i}$ and thus that the 
  resulting districts are $\ell$-fair.
  Moreover, the number of neighbors of $v_i$ in $H'$ is between $ \vert B_i \vert  - \smax +  \vert A_i\cup S_i \cup J_i \vert $ and $ \vert B_i \vert  - \max(0, \smin -  \vert A_i\cup S_i \cup J_i \vert )$, implying that there are at least $\smin- \vert A_i\cup S_i \cup J_i \vert $ and at most $\smax -  \vert A_i\cup S_i \cup J_i \vert $ vertices in $B_i$ that adjacent to some vertex from $I'$.
  Thus, the size of each district is in $[\smin, \smax]$. 
\end{proof}

Recall that  \fairCD is NP-hard and W[1]-hard with respect to $ \vert C \vert +k$ on trees with diameter four (\Cref{co:dia4}). In contrast to this, since a tree of diameter at most three has a vertex cover of size at most two, by \Cref{thm:xpvc}, the following holds.
\begin{corollary}\label{co:dia3}
  \fairCD is polynomial-time solvable on trees with diameter at most three.
\end{corollary}
Recall that we have shown in \Cref{co:dia4} that \fairCD is NP-hard on trees with diameter four.

Using a similar approach as for \Cref{thm:xpvc}, we show that \fairCD is fixed-parameter tractable with respect to $\vc +  \vert C \vert $.
The overall idea is the following:
We guess the partition of the vertex cover into the districts.
We categorize vertices outside the vertex cover according to their neighborhood and color.
We formulate the resulting problem as an ILP whose number of variables only depends on $\vc +  \vert C \vert $.
Subsequently we employ Lenstra's FPT-algorithm for ILPs \cite{Kan87,Len83} to 
obtain the following:
\begin{theorem}
\label{thm:fptvcc}
    \fairCD parameterized by $\vc+~ \vert C \vert $ is fixed-parameter tractable.
\end{theorem}
\begin{proof}
  Suppose there is a solution $\mathcal{V} = (V_1, \dots, V_k)$.
  Let $S$ be a vertex cover of minimum size.
  Then the vertices $I = V \setminus S$ that are not part of the vertex cover form an independent set.
  Let $k'$ be the number of districts in $\mathcal{V}$ that contain at least one vertex of $S$.
  By \Cref{obs:vc}, we have $k' \le \vc$.
  Assume without loss of generality that exactly for every $i \in [k']$, $S_i \ne \emptyset$ where $S_i = S \cap V_i$.
  Then, for every $i \in \{ k' + 1, \dots, k \}$, we have $V_i = \{ v \}$ where $v$ is some vertex from $I$.

  We say that a vertex $v \in I$ has type $(c, X)$ for $c \in C$ and $X \subseteq S$ if the color of $v$ in $G$ is $c$ and its neighborhood is $X$.
  Let $\mathcal{T}$ denote the set of all types (note that $ \vert \mathcal{T} \vert  = 2^{\vc} \cdot  \vert C \vert $).
  For type $T \in \mathcal{T}$, let $n_T$ be the number of vertices of type $T$, and for type~$T$ with $n_T > 0$ let $v_{T} \in I$ be an arbitrary vertex of type $T$.
  Our algorithm will construct an ILP instance for each choice of the following:
  \begin{itemize}
    \item
      An integer $k' \le \vc$.
      If $\ell = 0$ or $\smin > 1$, then we consider only one choice for $k'$, namely $k' = k$.
    \item
      A partition of $S$ into non-empty subsets $S_1, \dots, S_{k'}$.
      There are at most $\vc^{\vc}$ such partitions.
    \item
      For every $i \in [k']$, a set $\mathcal{T}_i$ of at most $ \vert S_i \vert  - 1$ types such that $G[S_i \cup J_i]$ is connected for $J_i = \{ v_T \mid T \in \mathcal{T}_i \}$. If $ \vert \{ i \mid T \in \mathcal{T}_i \} \vert  > n_T$ for some $T \in \mathcal{T}$, then we reject the current combination.
      Note that we can assume that $ \vert \mathcal{T}_i \vert  \le  \vert S_i \vert  - 1$  vertices are sufficient to make $G[S_i]$ connected by \Cref{obs:vc}.
      As there are at most $2^{\vc} \cdot  \vert C \vert $ types, the number of choices for all $\mathcal{T}_i$ is at most $\prod_{i \in [k']} (2^{\vc} \cdot  \vert C \vert )^{ \vert S_i \vert  - 1} \le 2^{\vc^2} \cdot  \vert C \vert ^{\vc}$.
    \item
      For every $i \in [k']$, a pair $(c_i, c_i')$ of colors, which have the largest numbers of occurrences in $V_i$ among all colors. In the sought solution it holds that $\cv_{c_i}(V_i) \ge \cv_{c_i'}(V_i)$.
      There are at most $ \vert C \vert ^{2\vc}$ choices for all pairs of colors.
  \end{itemize}
  To decide on the distribution of the vertices from $I$, we now construct an ILP. For the ILP formulation, we introduce an integer variable $x_{i, T}$ for each $i \in [k']$ and $T \in \mathcal{T}$.
  The variable $x_{i, T}$ will indicate the number of vertices of type $T$ that we put in $V_i$.
  Clearly, there are at most $n_T$ vertices for each type $T\in \mathcal{T}$ that belong to one of $V_1, \dots V_{k'}$:
  \begin{align*}
    \sum_{i \in [k']} x_{i, T} \le n_T.
  \end{align*}
  Then for every $i \in [k']$ and $T \in \mathcal{T}_i$, at least one vertex of type $T$ is included in $V_i$ to satisfy our previous guesses:
  \begin{align*}
    x_{i, T} \ge 1.
  \end{align*}
  Further for every $i \in [k']$, only vertices from $I$ that are adjacent to a vertex from $S_i$ can be part of~$V_i$. Thus, for all types $T=(c,X)\in \mathcal{T}$ with $X\cap S_i=\emptyset$ it holds that: 
  \begin{align*}
    x_{i, T} = 0.
  \end{align*}
  Moreover, there are exactly $k - k'$ vertices which are contained in none of $V_1, \dots, V_{k'}$ as they form their own districts:
  \begin{align*}
    \sum_{T \in \mathcal{T}} \left( n_T - \sum_{i \in [k']} x_{i, T} \right) = k - k'.
  \end{align*}
  For $i\in [k']$ and $c\in C$, let $n_{i, c}$ be the number of vertices of color $c$ in $V_i$:
  \begin{align*}
    n_{i, c} = \cv_{c}(S_i) + \sum_{T \in \mathcal{T} \text{ with } T = (c, X) \text{ for some }X\subseteq S} x_{i, T}.
  \end{align*}
  For every $i \in [k']$, the following two constraints will ensure that the district $V_i$ is $\ell$-fair.
  \begin{align*}
    n_{i,c_i'} \le n_{i, c_i} \le n_{i, c_i'} + \ell, \text{ and } n_{i, c_i'} \ge n_{i, c} \text{ for all } c \in C \setminus \{ c_i \}.
  \end{align*}
  Finally, for every $i \in [k']$, the following imposes the size constraints:
  \begin{align*}
    \smin \le \sum_{c \in C} n_{i, c} \le \smax.
  \end{align*}
  For the running time, observe that we construct at most $2^{\mathcal{O}(\vc^2)} \cdot  \vert C \vert ^{\vc}$ ILP instances.
  Since each ILP instance uses at most $\vc \cdot 2^{\vc} \cdot  \vert C \vert $ variables, 
it can be solved in $f(\vc +  \vert C \vert ) \cdot n^{\mathcal{O}(1)}$ time for some 
computable function $f$ due to a  result of Lenstra \cite{Kan87,Len83}.
\end{proof}

It remains open whether \fairCD is fixed-parameter tractable or W[1]-hard with respect to $\vc$. However, we can prove fixed-parameter tractability with respect to the number of vertices with degree at least two. 
Neglecting connected components consisting of two vertices, the set of vertices with degree at least two is always also a vertex cover; therefore, this parameter upper-bounds the vertex cover number in connected graphs on at least three vertices. 
The idea of the algorithm is to guess the partitioning of the degree-two vertices into the districts. 
Then, for each of these districts, there exists a set of degree-one vertices where each of these vertices either belongs to the district or needs to form its own. 
Lastly, we can distribute the degree-one vertices using \Cref{le:stars}. 
\begin{proposition}\label{pr:v2}
 Let $\nov$ be the number of vertices with degree at least 
 two. \fairCD is solvable in $\mathcal{O}(\nov^{\nov}\cdot k \cdot n)$ 
 time.
\end{proposition}
For the sake of simplicity, we assume that the graph does not contain connected components of size two (we can deal with them easily separately).
 Let $X\subseteq V$ be the set of vertices with degree at least two and $Y=V\setminus X$ be the set of degree-one vertices. 
 We iterate over all combinations of the following: 
 \begin{itemize}
  \item An integer $k'\leq k$. If $\ell=0$ or $\smin > 1$, then we only consider $k'=k$. 
  \item A partition of $X$ into $k'$ non-empty subsets $X_1,\dots, X_{k'}$. There are at most $\nov^{\nov}$ such partitions. 
 \end{itemize}
 For $i\in [k']$, let $Y_i\subseteq Y$ be the set of vertices from $Y$ that are adjacent to a vertex from $X_i$. 
 If $\ell=0$ or $\smin > 1$, then $k' = k$ and hence we accept if ($X_1\cup Y_1,\dots , X_k\cup Y_k$) is a solution and otherwise continue with the next combination.
 
Otherwise, we put all vertices from $X_i$ in district $V_i$. Each vertex $v\in Y_i$ is either 
part of~$V_i$ or forms its 
own district. Thereby, for each $i\in [k']$, we know by applying \Cref{le:stars}  an interval $[\alpha_i,\beta_i]$ for the number of $\ell$-fair districts respecting the size constraints in which the vertices $X_i\cup Y_i$ can be partitioned or we know that no solution exists (in this case, we continue with the next combination). We now check whether the given $k$ lies between $\sum_{i\in [k']} \alpha_i$ and $\sum_{i\in [k']} \beta_i$ and return the answer. 
  
 As there exist $\nov^{\nov}$ partitions of $X$, the running time of 
 the algorithm is $\mathcal{O}(\nov^{\nov} \cdot k\cdot n)$.

\section{Conclusion}
We initiated a thorough study of the NP-hard 
\textsc{Fair Connected Districting} (\fairCD) problem. 
We considered \fairCD on specific graph classes and analyzed the parameterized complexity of \fairCD with a focus on the number of districts, the number of colors, and various graph parameters. We have shown that while \fairCD can be solved on simple graph classes in polynomial time (mostly using approaches based on dynamic programming), on trees it is already NP-hard and W[1]-hard with respect to the combined parameter number of colors plus number of districts. Nevertheless, for graph parameters such as the vertex cover number and the max leaf number, we developed XP-algorithms. 

As to challenges for future research, 
we left open
whether \fairCD is fixed-parameter tractable or W[1]-hard with respect to the 
vertex cover number or max leaf number. 
The former question is particularly intriguing because it is closely related to 
an open question of Stoica et al.\ \cite{DBLP:conf/atal/StoicaCDG20} on the 
parameterized complexity of \textsc{Fair Regrouping} with respect to the number 
of districts. The latter is interesting because \fairCD could turn out to be 
one of few problems that are in XP yet W[1]-hard with respect to the max leaf 
number. 
It would be also promising to consider additional graph parameters (such as cliquewidth, treedepth or sparsity related parameters like the maximum degree or the degeneracy) or to examine further graph classes (such as grids).

From a broader perspective, there are several natural extensions of \fairCD. 
For instance, as already suggested by Stoica et 
 al.\ \cite{DBLP:conf/atal/StoicaCDG20}, there may be a function that specifies for each vertex a set of districts to which the vertex can belong.\footnote{However, it seems that in most applications where connectivity plays a crucial rule districts are ``isomorphic'' to each other, i.e., districts do not carry any particular meaning; for instance, when dividing a city into voting districts or people placed on a social network into teams for a competition. Thus, a priori, it is irrelevant for a vertex in which district it is put.}
 While all our hardness results still hold for this more complicated setting, it is open which of our algorithmic results can be adapted.
 Moreover, we did not study the generalization of \fairCD where each vertex has 
 an integer weight for each color (such a study has been done in the context of 
 gerrymandering by Cohen-Zemach 
 et al.\ \cite{DBLP:conf/atal/Cohen-ZemachLR18} and 
 Ito et al.\ \cite{DBLP:conf/atal/ItoK0O19}). 
 Weighted \fairCD is motivated by the following application scenarios: 
 First, in some settings we might be restricted to put a certain group of agents always in the same district (those can be combined into one vertex). 
 Second, if each vertex represents a voter and the colors represent alternatives, then voters might want to give points to different alternatives (such as done in the context of positional scoring rules). 

 Lastly, one can modify the definition of \fairCD. For instance, as done, among others, by Lewenberg et al.~\cite{DBLP:conf/atal/LewenbergLR17} and Eiben et al.~\cite{DBLP:conf/aaai/EibenFPS20} in the context of \textsc{Gerrymandering} and, among others, by Lu et al.~\cite{DBLP:conf/sigecom/LuSWZ10} and Fotakis and Tzamos~\cite{DBLP:journals/teco/FotakisT14} in the context of facility location problems, instead of placing agents on a social network, the agents may be placed 
 in a metric space. 
 The task is then to place $k$ ballot boxes in the space where each agent is assigned to the closest ballot box. 
The goal is again to make the resulting districts as fair as possible. 
\subsection*{Acknowledgments}
NB was supported by the DFG projects MaMu (NI 369/19) and ComSoc-MPMS (NI 369/22).
 TK was supported by the DFG projects FPTinP (NI 369/18) and DiPa (NI 369/21). This work
was started at the research retreat of the TU Berlin Algorithmics and Computational Complexity research group
held in Zinnowitz (Usedom) in September 2020.

\bibliographystyle{splncs04}
 \newcommand{\noop}[1]{}

\end{document}